\documentclass[12pt]{iopart}
\usepackage{iopams}
\usepackage{amsthm,dsfont,bm,amssymb}
\expandafter\let\csname equation*\endcsname\relax
\expandafter\let\csname endequation*\endcsname\relax
\usepackage{amsmath}
\usepackage{graphicx}
\usepackage{enumerate}
\usepackage[mathscr]{eucal}
\usepackage{tensor}
\usepackage{cite}

\newcommand{\bra}[1]{\mbox{$\langle #1 |$}}
\newcommand{\ket}[1]{\mbox{$| #1 \rangle$}}

\begin{document}

\review{Quantum Non-Markovianity: Characterization, Quantification and Detection}

\author{\'Angel Rivas$^{1}$, Susana F. Huelga$^{2,3}$ and Martin B. Plenio$^{2,3}$}
\address{$^1$Departamento de F\'isica Te\'orica I, Facultad de Ciencias F\'isicas, Universidad Complutense, 28040 Madrid, Spain.\\
$^2$Institut f\"ur Theoretische Physik, Universit\"at Ulm, Albert-Einstein-Allee 11, 89073 Ulm, Germany.\\
$^3$Center for Integrated Quantum Science and Technologies, Albert-Einstein-Allee 11, 89073 Ulm, Germany.}
\begin{abstract}
We present a comprehensive and up to date review on the concept of quantum non-Markovianity, a central theme in the theory of open quantum systems. We introduce the concept of quantum Markovian process as a generalization of the classical definition of Markovianity via the so-called divisibility property and relate this notion to the intuitive idea that links non-Markovianity with the persistence of memory effects. A detailed comparison with other definitions presented in the literature is provided. We then discuss several existing proposals to quantify the degree of non-Markovianity of quantum dynamics and to witness non-Markovian behavior, the latter providing sufficient conditions to detect deviations from strict Markovianity. Finally, we conclude by enumerating some timely open problems in the field and provide an outlook on possible research directions.
\end{abstract}

\maketitle

\newtheorem{theorem}{Theorem}[section]
\newtheorem{prop}{Proposition}[section]
\newtheorem{cor}{Corollary}[section]
\theoremstyle{definition}
\newtheorem{defi}{Definition}[section]
\newtheorem{ex}{Example}[section]

\tableofcontents
\pagestyle{plain}

\section{Introduction}
In recent years, renewed attention has been paid to the characterization of quantum non-Markovian processes. Different approaches have been followed and several methods have been proposed which in some cases yield inequivalent conclusions. Given the considerable amount of literature that has built up on the subject, we believe that the time is right to summarize  most of the existing results in a review article that clarifies both the underlining structure and the interconnections between the different approaches.

On the one hand, we are fully aware of the risk we take by writing a review on a quite active research field, with new results continuously arising during the writing of this work. We do hope, on the other hand, that possible shortcomings will be well balanced by the potential usefulness of such a review in order to, hopefully, clarify some misconceptions and generate further interest in the field.

Essentially, the subject of quantum non-Markovianity addresses two main questions, namely:
\begin{enumerate}[I.]
\item What is a quantum Markovian process and hence what are non-Markovian processes? (\emph{characterization problem}).
\item If a given process deviates from Markovianity, by how much does it deviate? (\emph{quantification problem}).
\end{enumerate}

In this work we examine both questions in detail. More specifically, concerning the characterization problem, we adopt the so-called called \emph{divisibility property} as a definition of \emph{quantum Markovian processes}. As this is not the only approach to non-Markovianity, in Section \ref{sec:comparison} we introduce and discuss other proposed definitions and compare them to the divisibility approach. In this regard, we  would like to stress that it is neither our intention nor is the field at a stage that allows us to decide on a definitive definition of quantum Markovian processes. It is our hope however that we will convince the reader that the strong analogy between the definition for non-Markovianity taken in this work and the classical definition of Markov processes, and the ensuing good mathematical properties which will allow us to address the characterization problem in simple terms, represents a fruitful approach to the topic. Concerning the quantification problem, we discuss most of the quantifiers present in the literature, and we classify them into measures and witnesses of non-Markovianity, depending on whether they are able to detect every non-Markovian dynamics or just a subset. Given the large body of literature that explores the application of these methods to different physical realizations, we have opted for keeping the presentation mainly on the abstract level and providing a detailed list of references. However, we have also included some specific examples for the sake of illustration of fundamental concepts.

This work is organized as follows. In Section 2, we recall the classical concept of Markovian process and some of its main properties. This is crucial in order to understand why the divisibility property provides a good definition of quantum Markovianity. In Section 3 we introduce the concept of quantum Markovian process by establishing a step by step parallelism with the classical definition, and explain in detail why these quantum processes can be considered as memoryless. Section 4 gives a detailed review of different measures of non-Markovianity and Section 4 describes different approaches in order to construct witnesses able to detect non-Markovian dynamics. Finally, Section 6 is devoted to conclusions and to outline some of the problems which remain open in this field and possible future research lines.

\section{Markovianity in Classical Stochastic Processes}\label{sec-ClassicalMarkov}

In order to give a definition of a Markov process in the quantum regime, it is essential to understand the concept of Markov process in the classical setting. Thus, this section is devoted to revise the definition of classical Markov processes and sketch the most interesting properties for our purposes without getting too concerned with mathematical rigor. More detailed explanations on the foundations of stochastic processes can be found in \cite{Doob90,Gardiner97,Norris97,Parzen99,Ethier05,vanKampen07}.

\subsection{Definition and properties}
Consider a random variable $X$ defined on a classical probability space $(\Omega,\Sigma,\mathbb{P})$, where $\Omega$ is a given set (the sample space), $\Sigma$ (the possible events) is a $\sigma-$algebra of subsets of $\Omega$, containing $\Omega$ itself, and the probability $\mathbb{P}:\Sigma\rightarrow [0,1]$ is a $\sigma-$additive function with the property that $\mathbb{P}(\Omega)=1$, (cf. \cite{Doob90,Gardiner97,Norris97,Parzen99,Ethier05,vanKampen07}). In order to avoid further problems when considering conditional probabilities (see for example the Borel-Kolmogorov paradox \cite{Kol}) we shall restrict attention from now on to discrete random variables, i.e. random variables which take values on a finite set denoted by $\mathcal{X}$.

A \emph{classical stochastic process} is a family of random variables $\{X(t),t\in I\subset\mathds{R}\}$. Roughly speaking, this is nothing but a random variable $X$ depending on a parameter $t$ which usually represents time. The story starts with the following definition.

\begin{defi}[Markov process]
A stochastic process $\{X(t),t\in I\}$ is a \emph{Markov process} if the probability that the random variable $X$ takes a value $x_n$ at any arbitrary time $t_n\in I$, provided that it took the value $x_{n-1}$ at some previous time $t_{n-1}<t_{n}\in{I}$, is uniquely determined, and not affected by the possible values of $X$ at previous times to $t_{n-1}$. This is formulated in terms of conditional probabilities as follows
\begin{equation}\label{Markov}
 \mathbb{P}(x_n,t_n|x_{n-1},t_{n-1};\ldots;x_0,t_0)=\mathbb{P}(x_n,t_n|x_{n-1},t_{n-1}),\quad \mathrm{for\ all\ }\{t_n\geq t_{n-1}\geq\ldots\geq t_0\}\subset I,
\end{equation}
\end{defi}
\noindent and informally it is summarized by the statement that ``a Markov process does not have memory of the history of past values of $X$''. This kind of stochastic processes are named after the Russian mathematician A. Markov \cite{footnote1}.

From the previous definition \eqref{Markov} it is possible to work out further properties of Markov processes. For instance, it follows immediately from \eqref{Markov} that for a Markov process
\begin{equation}\label{ConditionalMarkov}
\mathbb{E}(x_n,t_n|x_{n-1},t_{n-1};\ldots;x_0,t_0)=\mathbb{E}(x_n,t_n|x_{n-1},t_{n-1}),\quad \mathrm{for\ all\ }\{t_n\geq t_{n-1}\geq\ldots\geq t_0\}\subset I,
\end{equation}
where $\mathbb{E}(x|y)=\sum_{x\in\mathcal{X}}x\mathbb{P}(x|y)$ denotes the so-called \emph{conditional expectation}.

In addition, Markov processes satisfy another remarkable property. If we take the joint probability for any three consecutive times $t_3>t_2>t_1$ and apply the definition of conditional probability twice we obtain
\begin{eqnarray}\label{Chap-Kolmo0}
\mathbb{P}(x_3,t_3;x_2,t_2;x_1,t_1)&=&\mathbb{P}(x_3,t_3|x_2,t_2;x_1,t_1)\mathbb{P}(x_2,t_2;x_1;t_1)\nonumber\\
&=&\mathbb{P}(x_3,t_3|x_2,t_2;x_1,t_1)\mathbb{P}(x_2,t_2|x_1,t_1)\mathbb{P}(x_1,t_1).
\end{eqnarray}
Since the Markov condition \eqref{Markov} implies that $\mathbb{P}(x_3,t_3|x_2,t_2;x_1,t_1)=\mathbb{P}(x_3,t_3|x_2,t_2)$, by taking the sum over $x_2$ and dividing both sides by $\mathbb{P}(x_1,t_1)$ we arrive at
\begin{equation}\label{Chap-Kolmo}
\mathbb{P}(x_3,t_3|x_1,t_1)=\sum_{x_2\in\mathcal{X}} \mathbb{P}(x_3,t_3|x_2,t_2)\mathbb{P}(x_2,t_2|x_1,t_1),
\end{equation}
which is called the \emph{Chapman-Kolmogorov equation}. Moreover, the next theorem gives an answer to the converse statement.

\begin{theorem}\label{Theo-CK-Markov}
A family of conditional probabilities $\mathbb{P}(x_n,t_n|x_{n-1},t_{n-1})$ with $t_{n}>t_{n-1}$ satisfying \eqref{Chap-Kolmo} can always be seen as the conditional probabilities of a Markov process $\{X(t),t\in I\}$.
\end{theorem}
\begin{proof}
The proof is by construction. Take some probabilities $\mathbb{P}(x_n,t_n)$ and define the two-point joint probabilities by
\[
\mathbb{P}(x_n,t_n;x_{n-1},t_{n-1}):=\mathbb{P}(x_n,t_n|x_{n-1},t_{n-1})\mathbb{P}(x_{n-1},t_{n-1}).
\]
Then, set
\begin{equation}
\mathbb{P}(x_n,t_n|x_{n-1},t_{n-1};\ldots;x_0,t_0):=\mathbb{P}(x_n,t_n|x_{n-1},t_{n-1}),\quad \mathrm{for\ all\ \ }\{t_n\geq t_{n-1}\geq\ldots\geq t_0\}\subset I.
\end{equation}
and construct higher joint probabilities by using expressions analogous to Eq. \eqref{Chap-Kolmo0}. This construction is always possible as it is compatible with \eqref{Chap-Kolmo}, which is the presupposed condition satisfied by $\mathbb{P}(x_n,t_n|x_{n-1},t_{n-1})$.
\end{proof}

\subsection{Transition matrices}

In this section we shall focus on the evolution of one-point probabilities $\mathbb{P}(x,t)$ during a stochastic process. Thus, consider a linear map $T$ that connects the probability of a random variable $X$, at different times $t_0$ and $t_1$:
\begin{equation}
\mathbb{P}(x_1,t_1)=\sum_{x_0\in\mathcal{X}}T(x_1,t_1|x_0,t_0)\mathbb{P}(x_0,t_0).
\end{equation}
Since $\sum_{x_1\in\mathcal{X}}\mathbb{P}(x_1,t_1)=1$ and $\mathbb{P}(x_1,t_1)\geq0$ for every $\mathbb{P}(x_0,t_0)$, we conclude that
\begin{align}
&\sum_{x_1\in\mathcal{X}}T(x_1,t_1|x_0,t_0)=1,\\
&T(x_1,t_1|x_0,t_0)\geq0, \quad x_1,x_0\in\mathcal{X}.
\end{align}
Matrices $T$ fulfilling these properties are called \emph{stochastic matrices}.

Consider $t=t_0$ to be the initial time of some (not necessarily Markovian) stochastic process $\{X(t),t\in I\}$. From the definition of conditional probability,
\begin{equation}
\mathbb{P}(x_2,t_2;x_0,t_0)=\mathbb{P}(x_2,t_2|x_0,t_0)\mathbb{P}(x_0,t_0)\Rightarrow\mathbb{P}(x_2,t_2)=\sum_{x_0\in\mathcal{X}}\mathbb{P}(x_2,t_2|x_0,t_0)\mathbb{P}(x_0,t_0),
\end{equation}
and therefore $T(x_2,t_2|x_0,t_0)=\mathbb{P}(x_2,t_2|x_0,t_0)$ for every $t_2$. This relation is not valid in general for $t_1>t_0$, $T(x_2,t_2|x_1,t_1)\neq\mathbb{P}(x_2,t_2|x_1,t_1)$. The reason is that $\mathbb{P}(x_2,t_2|x_1,t_1)$ is not fully defined for a general stochastic process; we need to know the value of $X$ for previous time instants as $\mathbb{P}(x_2,t_2|x_1,t_1;x_0,t_0)$ could be different from $\mathbb{P}(x_2,t_2|x_1,t_1;x'_0,t_0)$ for $x_0\neq x'_0$. However that is not the case for Markov processes which satisfy the following result.

\begin{theorem}
Consider a Markov process $\{X(t),t\in I\}$. Given any two time instants $t_1$ and $t_2$ we have
\begin{equation}
T(x_2,t_2|x_1,t_1)=\mathbb{P}(x_2,t_2|x_1,t_1).
\end{equation}
\end{theorem}
\begin{proof} It follows from the fact that we can write
$\mathbb{P}(x_2,t_2;x_1,t_1)=\mathbb{P}(x_2,t_2|x_1,t_1)\mathbb{P}(x_1,t_1)$, as $\mathbb{P}(x_2,t_2|x_1,t_1)$ is well defined for any $t_1$ and $t_2$.
\end{proof}

\noindent From this theorem and the Chapman-Kolmogorov equation \eqref{Chap-Kolmo} we obtain the following corollary.
\begin{cor} Consider a Markov process $\{X(t),t\in I\}$, then for any $t_3\geq t_2\geq t_1\geq t_0$, the transition matrix satisfies the properties
\begin{align}
&\sum_{x_2\in\mathcal{X}}T(x_2,t_2|x_1,t_1)=1,\label{divisible1}\\
&T(x_2,t_2|x_1,t_1)\geq0, \label{divisible2}\\
&T(x_3,t_3|x_1,t_1)=\sum_{x_2\in\mathcal{X}}T(x_3,t_3|x_2,t_2)T(x_2,t_2|x_1,t_1)\label{divisible3}.
\end{align}
\end{cor}

In summary, for a Markov process the transition matrices are the two-point conditional probabilities and satisfy the composition law Eq. \eqref{divisible3}. Essentially, Eq. \eqref{divisible3} states that the evolution from $t_1$ to $t_3$ can be written as the composition of the evolution from $t_1$ to some intermediate time $t_2$, and from this $t_2$ to the final time $t_3$.

In case of non-Markovian processes, $T(x_2,t_2|x_1,t_1)$ might be not well defined for $t_1\neq t_0$. Nevertheless, if the matrix $\mathbb{P}(x_1,t_1|x_0,t_0)$ is invertible for every $t_1$, then $T(x_2,t_2|x_1,t_1)$ can be written in terms of well-defined quantities. Since the evolution from $t_1$ to $t_2$ (if it exists) has to be the composition of the backward evolution to the initial time $t_0$ and the forward evolution from $t_0$ to $t_2$, we can write
\begin{eqnarray}
T(x_2,t_2|x_1,t_1)&=&\sum_{x_0\in\mathcal{X}}T(x_2,t_2|x_0,t_0)T(x_0,t_0|x_1,t_1)\nonumber\\
&=&\sum_{x_0\in\mathcal{X}}\mathbb{P}(x_2,t_2|x_0,t_0)[\mathbb{P}(x_1,t_1|x_0,t_0)]^{-1}.
\end{eqnarray}
In this case the composition law Eq. \eqref{divisible3} is satisfied and Eq. \eqref{divisible1} also holds. However, condition Eq. \eqref{divisible2} may be not fulfilled, which prevents any interpretation of $T(x_2,t_2|x_1,t_1)$ as a conditional probability and therefore manifests the non-Markovian character of such a stochastic process.

\begin{defi}[Divisible process] A stochastic process $\{X(t),t\in I\}$ for which the associates transition matrices satisfy Eqs. \eqref{divisible1}, \eqref{divisible2} and \eqref{divisible3} is called \emph{divisible}. \label{Def:clas-divisible}
\end{defi}

There are divisible processes which are non-Markovian. As an example (see \cite{Parzen99,vanKampen07}), consider a stochastic process $\{X(t),t\in I\}$ with two possible results $\mathcal{X}=\{0,1\}$, and just three discrete times  $I=\{t_1,t_2,t_3\}$ ($t_3>t_2>t_1$). Define the joint probabilities as
\begin{multline}
\mathbb{P}(x_3,t_3;x_2,t_2;x_1,t_1):=\frac{1}{4}\left(\delta_{x_3,0}\delta_{x_2,0}\delta_{x_1,1}+\delta_{x_3,0}\delta_{x_2,1}\delta_{x_1,0}\right.\\
+\left.\delta_{x_3,1}\delta_{x_2,0}\delta_{x_1,0}+\delta_{x_3,1}\delta_{x_2,1}\delta_{x_1,1}\right).
\end{multline}
By computing the marginal probabilities we obtain $\mathbb{P}(x_3,t_3;x_2,t_2)=\mathbb{P}(x_2,t_2;x_1,t_1)=\mathbb{P}(x_3,t_3;x_1,t_1)=1/4$, and then
\begin{multline}
\mathbb{P}(x_3,t_3|x_2,t_2;x_1,t_1)=\frac{\mathbb{P}(x_3,t_3;x_2,t_2;x_1,t_1)}{\mathbb{P}(x_2,t_2;x_1,t_1)}=\left(\delta_{x_3,0}\delta_{x_2,0}\delta_{x_1,1}\right.\\
+\left. \delta_{x_3,0}\delta_{x_2,1}\delta_{x_1,0}+\delta_{x_3,1}\delta_{x_2,0}\delta_{x_1,0}+\delta_{x_3,1}\delta_{x_2,1}\delta_{x_1,1}\right).
\end{multline}
Therefore the process is non-Markovian as, for example, $\mathbb{P}(1,t_3|0,t_2;0,t_1)=1$ and  $\mathbb{P}(1,t_3|0,t_2;1,t_1)=0$. However the transition matrices can be written as
\[
T(x_3,t_3|x_2,t_2)=\frac{\mathbb{P}(x_3,t_3;x_2,t_2)}{\mathbb{P}(x_2,t_2)}=\frac{1}{2},
\]
and similarly $T(x_2,t_2|x_1,t_1)=T(x_3,t_3|x_1,t_1)=1/2$. Hence the conditions \eqref{divisible1}, \eqref{divisible2} and \eqref{divisible3} are clearly fulfilled. Other examples of non-Markovian divisible processes can be found in \cite{Levi49,Feller59,Vacch-NJP,Vacch-JPB,Vacch-Class}.

Despite the existence non-Markovian divisible processes, we can establish the following key theorem.

\begin{theorem} \label{TheoTransMatrixMarko} A family of transition matrices $T(x_2,t_2|x_1,t_1)$ with $t_2>t_1$ which satisfies Eqs. \eqref{divisible1}, \eqref{divisible2} and \eqref{divisible3} can always be seen as the transition matrices of some underlying Markov process $\{X(t),t\in I\}$.
\end{theorem}
\begin{proof}
Since the matrices $T(x_2,t_2|x_1,t_1)$ satisfy \eqref{divisible1} and \eqref{divisible2}, they can be understood as conditional probabilities $\mathbb{P}(x_2,t_2|x_1,t_1)=T(x_2,t_2|x_1,t_1)$, and since \eqref{divisible3} is also satisfied, the process fulfils Eq. \eqref{Chap-Kolmo}. Then the final statement follows from Theorem \ref{Theo-CK-Markov}.
\end{proof}
\noindent Thus, we conclude that:
\begin{cor} At the level of one-point probabilities, divisible and Markovian processes are equivalent. The complete hierarchy of time-conditional probabilities has to be known to make any distinctions. \label{Cor-Divisible}
\end{cor}

\subsection{Contractive property}\label{sec:contractiveClassical}

There is another feature of Markov processes that will be useful in the quantum case. Consider a vector $v(x)$, where $x$ denotes its different components. Then its \emph{$L_1-$norm} is defined as
\begin{equation}
\|v(x)\|_1:=\sum_{x}|v(x)|.
\end{equation}

This norm is particularly useful in hypothesis testing problems. Namely, consider a random variable $X$ which is distributed according to either probability $p_1(x)$ or probability $p_2(x)$. We know that, with probability $q$, $X$ is distributed according to $p_1(x)$, and, with probability $1-q$, $X$ is distributed according to $p_2(x)$. Our task consists of sampling $X$ just once with the aim of inferring the correct probability distribution of $X$ [$p_1(x)$ or $p_2(x)$]. Then the minimum (averaged) probability to give the wrong answer turns out to be
\begin{equation}
\mathbb{P}_{\rm min}({\rm fail})=\frac{1-\|w(x)\|_1}{2},
\end{equation}
where $w(x):=qp_1(x)-(1-q)p_2(x)$. The proof of this result follows the same steps as in the quantum case (see Section \ref{sectionQcontraction}). Thus the $L_1$-norm of the vector $w(x)$ gives the capability to distinguish correctly between $p_1(x)$ and $p_2(x)$ in the two-distribution discrimination problem.

Particularly, in the unbiased case $q=1/2$, we have
\[
\|w(x)\|_1=\frac{1}{2}\|p_1(x)-p_2(x)\|_1,
\]
which is known as the \emph{Kolmogorov distance}, \emph{$L_1$-distance}, or \emph{variational distance} between $p_1(x)$ and $p_2(x)$.

In the identification of non-divisible processes, the $L_1$-norm also plays an important role.
\begin{theorem} \label{theoremClassi-contraction} Let $T(x_2,t_2|x_1,t_1)$ be the transition matrices of some stochastic process. Then such a process is divisible if and only if the $L_1$-norm does not increase when $T(x_2,t_2|x_1,t_1)$ is applied to every vector $v(x)$, $x\in\mathcal{X}$, for all $t_1$ and $t_2$,
\begin{equation}\label{Classic-contraction}
\left\|\sum_{x_1\in\mathcal{X}}T(x_2,t_2|x_1,t_1)v(x_1)\right\|_1\leq\| v(x_2)\|_1, \quad t_1\leq t_2.
\end{equation}
\end{theorem}
\begin{proof} The ``only if'' part follows from the properties \eqref{divisible1} and \eqref{divisible2}:
\begin{eqnarray}
\left\|\sum_{x_1\in\mathcal{X}}T(x_2,t_2|x_1,t_1)v(x_1)\right\|_1&=&\sum_{x_2\in\mathcal{X}}\left|\sum_{x_1\in\mathcal{X}}T(x_2,t_2|x_1,t_1)v(x_1)\right|\nonumber\\
&\leq&\sum_{x_1,x_2\in\mathcal{X}}T(x_2,t_2|x_1,t_1)|v(x_1)|\nonumber\\
&=&\sum_{x_1\in\mathcal{X}}|v(x_1)|=\sum_{x_2\in\mathcal{X}}|v(x_2)|=\| v(x_2)\|.
\end{eqnarray}
For the ``if'' part, as we mentioned earlier, if $T(x_2,t_2|x_1,t_1)$ does exist, it always satisfies Eqs. \eqref{divisible1} and \eqref{divisible3}. Take a vector to be a probability distribution $v(x)=p(x)\geq0$ for all $x\in\mathcal{X}$, because of Eq. \eqref{divisible1} we have
\begin{equation}
\| p(x_1)\|_1=\sum_{x_1\in\mathcal{X}}p(x_1)=\sum_{x_1,x_2\in\mathcal{X}}T(x_2,t_2|x_1,t_1)p(x_1).
\end{equation}
Since, by hypothesis, Eq. \eqref{Classic-contraction} holds for any vector, we obtain the following chain of inequalities
\begin{eqnarray}
 \| p(x_1)\|_1&=&\sum_{x_1,x_2\in\mathcal{X}}T(x_2,t_2|x_1,t_1)p(x_1)\leq\sum_{x_2\in\mathcal{X}}\left|\sum_{x_1\in\mathcal{X}}T(x_2,t_2|x_1,t_1)p(x_1)\right|\nonumber\\
&\leq&\sum_{x_2\in\mathcal{X}}|p(x_2)|=\sum_{x_1\in\mathcal{X}}|p(x_1)|=\| p(x_1)\|.
\end{eqnarray}
Therefore,
\[
\sum_{x_2\in\mathcal{X}}\left|\sum_{x_1\in\mathcal{X}}T(x_2,t_2|x_1,t_1)p(x_1)\right|=\sum_{x_1,x_2\in\mathcal{X}}T(x_2,t_2|x_1,t_1)p(x_1),
\]
for any probability $p(x_1)$, which is only possible if $\sum_{x_1\in\mathcal{X}}T(x_2,t_2|x_1,t_1)p(x_1)\geq0$. Then Eq. \eqref{divisible2} has to be satisfied.
\end{proof}

Because of this theorem and Eq. \eqref{divisible3}, $\mathbb{P}_{\rm min}({\rm fail})$ increases monotonically with time for a divisible process. In this regard, if the random variable $X$ undergoes a Markovian process, the best chance to rightly distinguish between the two possible distributions $p_1(x)$ and $p_2(x)$ is to sample $X$ at time instants as close as possible to the initial time $t_0$. However that is not the case if $X$ is subject to a non-divisible (and then non-Markovian) process. Then, in order to decrease the error probability, it could be better to wait until some time, say $t_1$, where $\left\|w(x_1,t_1)\right\|_1=\left\|qp_1(x_1,t_1)-(1-q)p_2(x_1,t_1)\right\|_1$ increases again (without exceeding its initial value). The fact that the error probability may decrease for some time $t_1$ after the initial time $t_0$ can be understood as a trait of underlying memory in the process. That is, the system retains some information about the probability of $X$ at $t_0$, which arises at a posterior time in the evolution.

In summary, classical Markovian processes are defined via multi-time conditional probabilities,  Eq. \eqref{Markov}. However, if the experimenter only has access to one-point probabilities, Markovian processes become equivalent to divisible processes. The latter are more easily characterized, as they only depend on properties of transition matrices and the $L_1$-norm.

\section{Markovianity in Quantum Processes}\label{sec-QuantumMarkov}

After the succinct review of classical Markovian processes in the previous section, here we shall try to adapt those concepts to the quantum case. By the adjective ``quantum'' we mean that the system undergoing evolution is a quantum system. Our aim is to find a simple definition of a quantum Markovian process by keeping a close analogy to its classical counterpart. Since this is not straightforward, we comment first on some points which make a definition of quantum Markovianity difficult to formulate. For the sake of simplicity, in the following we shall consider finite dimensional quantum systems unless otherwise stated.

\subsection{Problems of a straightforward definition}
Since the quantum theory is a statistical theory, it seems meaningful to ask for some analogue to classical stochastic processes and particularly Markov processes. However, the quantum theory is based on non-commutative algebras and this makes its analysis considerably more involved. Indeed, consider the classical definition of Markov process Eq. \eqref{Markov}, to formulate a similar condition in the quantum realm we demand a way to obtain $\mathbb{P}(x_n,t_n|x_{n-1},t_{n-1};\ldots;x_0,t_0)$ for quantum systems. The problem arises because we can sample a classical random variable without affecting its posterior statistics; however, in order to ``sample'' a quantum system, we need to perform measurements, and these measurements disturb the state of the system and affect the subsequent outcomes. Thus, $\mathbb{P}(x_n,t_n|x_{n-1},t_{n-1};\ldots;x_0,t_0)$ does not only depend on the dynamics but also on the measurement process, and a definition of quantum Markovianity in terms of it, even if possible, does not seem very appropriate. Actually, in such a case the Markovian character of a quantum dynamical system would depend on which measurement scheme is chosen to obtain $\mathbb{P}(x_n,t_n|x_{n-1},t_{n-1};\ldots;x_0,t_0)$. This is very inconvenient as the definition of Markovianity should be independent of what is required to verify it.

\subsection{Definition in terms of one-point probabilities: divisibility}\label{NonMarkoSection}
Given the aforementioned problems to construct $\mathbb{P}(x_n,t_n|x_{n-1},t_{n-1};\ldots;x_0,t_0)$ in the quantum case, a different approach focuses on the study of one-time probabilities $\mathbb{P}(x,t)$. For these, the classical definition of Markovianity reduces to the concept of divisibility (see Definition \ref{Def:clas-divisible}), and a very nice property is that divisibility may be defined in the quantum case without any explicit mention to measurement processes. To define quantum Markovianity in terms of divisibility may seem to lose generality, nevertheless Theorem \ref{TheoTransMatrixMarko} and Corollary \ref{Cor-Divisible} assert that this loss is innocuous, as divisibility and Markovianity are equivalent properties for one-time probabilities. These probabilities are the only ones that can be constructed in the quantum case avoiding the difficulties associated to measurement disturbance.

Let us consider a system in a quantum state given by some (non-degenerate) density matrix $\rho$, the spectral decomposition yields
\begin{equation}
\rho=\sum_x p(x)|\psi(x)\rangle\langle\psi(x)|.
\end{equation}
Here the eigenvalues $p(x)$ form a classical probability distribution, which may be interpreted as the probabilities for the system to be in the corresponding eigenstate $|\psi(x)\rangle$,
\begin{equation}
\mathbb{P}(\ket{\psi(x)}):=p(x).
\end{equation}
Consider now some time evolution of the quantum system such that the spectral decomposition of the initial state is preserved; $\rho(t_0)=\sum_x p(x,t_0)|\psi(x)\rangle\langle\psi(x)|$ is mapped to
\begin{equation}\label{spectral(t)}
\rho(t)=\sum_x p(x,t)|\psi(x)\rangle\langle\psi(x)|\in\mathcal{S},
\end{equation}
where $\mathcal{S}$ denotes the set of quantum states with the same eigenvectors as $\rho(t_0)$. Since this process can be seen as a classical stochastic process on the variable $x$, which labels the eigenstates $|\psi(x)\rangle$, we consider it to be divisible if the evolution of $p(x,t)$ satisfies the classical definition of divisibility (Definition \ref{Def:clas-divisible}). In such a case, there are transition matrices $T(x_1,t_1|x_0,t_0)$, such that
\begin{equation}\label{spectralTransMatrix}
p(x_1,t_1)=\sum_{x_0\in\mathcal{X}}T(x_1,t_1|x_0,t_0)p(x_0,t_0),
\end{equation}
fulfilling Eqs. \eqref{divisible1}, \eqref{divisible2} and \eqref{divisible3}. This Eq. \eqref{spectralTransMatrix} can be written in terms of density matrices as
\begin{equation}
\rho(t_1)=\mathcal{E}_{(t_1,t_0)}\left[\rho(t_0)\right].
\end{equation}
Here, $\mathcal{E}_{(t_1,t_0)}$ is a dynamical map that preserves the spectral decomposition of $\rho(t_0)$ and satisfies
\begin{align}
\mathcal{E}_{(t_1,t_0)}\left[\rho(t_0)\right]&=\sum_{x_0 \in\mathcal{X}}p(x_0,t_0) \mathcal{E}_{(t_1,t_0)}[|\psi(x_0)\rangle\langle\psi(x_0)|]\nonumber\\
&=\sum_{x_1,x_0 \in\mathcal{X}}T(x_1,t_1|x_0,t_0)p(x_0,t_0)|\psi(x_1)\rangle\langle\psi(x_1)|.
\end{align}
Furthermore, because of Eqs. \eqref{divisible1}, \eqref{divisible2} and \eqref{divisible3}, $\mathcal{E}_{(t_2,t_1)}$ preserves positivity and the trace of any state in $\mathcal{S}$ and obeys the composition law
\begin{equation}\label{CompositionLaw}
\mathcal{E}_{(t_3,t_1)}=\mathcal{E}_{(t_3,t_2)}\mathcal{E}_{(t_2,t_1)}, \quad t_3\geq t_2 \geq t_1.
\end{equation}

On the other hand, since the maps $\{\mathcal{E}_{(t_2,t_1)},t_2\geq t_1\geq t_0\}$ are supposed to describe some quantum evolution, they are linear (there is not experimental evidence against this fact \cite{Zeilinger,Weinberg,Wineland}). Thus, their action on another set $\mathcal{S}'$ of quantum states with different spectral projectors to $\mathcal{S}$ is physically well defined provided that the positivity of the states of $\mathcal{S}'$ is preserved (i.e. any density matrix in $\mathcal{S}'$ is transformed in another valid density matrix). Hence, by consistence, we formulate the following general definition of a P-divisible process.

\begin{defi}[P-divisible process] \label{defpre-Marko2} We say that a quantum system subject to some time evolution characterized by the family of trace-preserving linear maps $\{\mathcal{E}_{(t_2,t_1)},t_2\geq t_1\geq t_0\}$ is \emph{P-divisible} if, for every $t_2$ and $t_1$, $\mathcal{E}_{(t_2,t_1)}$ is a positive map (preserve the positivity of any quantum state) and fulfils Eq. \eqref{CompositionLaw}.
\end{defi}

The reason to use the terminology ``P-divisible'' (which stands for positive-divisible) instead of ``divisible'' comes from the difference between positive and complete positive maps which is essential in quantum mechanics. More explicitly, a linear map $\Upsilon$ acting on a matrix space $\mathcal{M}$ is a positive map if for $A\in \mathcal{M}$,
\begin{equation}
A\geq0 \Rightarrow \Upsilon(A)\geq0,
\end{equation}
i.e. $\Upsilon$ transforms positive semidefinite matrices into positive semidefinite matrices. In addition, $\Upsilon$ is said to be completely positive if for any matrix space $\mathcal{M}'$ such that $\dim(\mathcal{M}')\geq\dim(\mathcal{M})$, and $B\in\mathcal{M}'$,
\begin{equation}\label{CP}
B\geq0 \Rightarrow \Upsilon\otimes\mathds{1}(B)\geq0.
\end{equation}
These concepts are properly extended to the infinity dimensional case \cite{Paulsen}.

Complete positive maps are much easier to characterize than maps that are merely positive \cite{Choi,Jamiolkowski}; they admit the so-called Kraus representation, $\Upsilon(\cdot)=\sum_j K_j(\cdot)K_j^\dagger$, and it can be shown  that if Eq. \eqref{CP} is fulfilled with $\dim(\mathcal{M}')=\dim(\mathcal{M})^2$, it is also true for any $\mathcal{M}'$ such that $\dim(\mathcal{M}')\geq\dim(\mathcal{M})$.

It is well-know that the requirement of positivity alone for a dynamical map presents difficulties. Concretely, in order to keep the positivity of density matrices in presence of entanglement with another extra system we must impose complete positivity instead of positivity \cite{NC00,Kraus71,Daviesbook76,Kraus83,BrPe02,AlickiLendi87,Libro}. Thus, now we are able to give a definition of quantum Markovian process.

\begin{defi}[Markovian quantum process] \label{def-Marko} We shall say that a quantum system subject to a time evolution given by some family of trace-preserving linear maps $\{\mathcal{E}_{(t_2,t_1)},t_2\geq t_1\geq t_0\}$ is \emph{Markovian} (or  \emph{divisible} \cite{Wolf1}) if, for every $t_2$ and $t_1$, $\mathcal{E}_{(t_2,t_1)}$ is a complete positive map and fulfills the composition law Eq. \eqref{CompositionLaw}.
\end{defi}

For the sake of comparison, the following table shows the clear parallelism between classical transition matrices and quantum evolution families in a Markovian process.

\begin{table}[h!]
\begin{tabular}{|l|c|c|}
\cline{2-3}
    \multicolumn{1}{l|}{} & Classical & Quantum \\
\hline
{\small Normalization}   & {\footnotesize $\sum_{x_2\in\mathcal{X}}T(x_2,t_2|x_1,t_1)=1$} & {\footnotesize $\mathcal{E}_{(t_2,t_1)}$ trace-preserving} \\
\hline
{\small Positivity}   & {\footnotesize $T(x_2,t_2|x_1,t_1)\geq0$} & {\footnotesize $\mathcal{E}_{(t_2,t_1)}$ completely positive} \\
\hline
{\small Composition Law}   & {\scriptsize $T(x_3,t_3|x_1,t_1)=\sum_{x_2\in\mathcal{X}}T(x_3,t_3|x_2,t_2)T(x_2,t_2|x_1,t_1)$} & {\footnotesize $\mathcal{E}_{(t_3,t_1)}=\mathcal{E}_{(t_3,t_2)}\mathcal{E}_{(t_2,t_1)}$} \\
\hline
\end{tabular}
\end{table}

Before we move on, it is worth to summarize the argument leading to the definition of Markovian quantum process, as it is the central concept of this work. Namely, since a direct definition from the classical condition Eq. \eqref{Markov} is problematic because of quantum measurement disturbance, we focus on one-time probabilities. For those, classical Markovian processes and divisible processes are equivalent, thus we straightforward formulate the divisibility condition for quantum dynamics preserving the spectral decomposition of certain set of states $\mathcal{S}$. Then the Markovian (or divisibility) condition for any quantum evolution follows by linearity when taking into account the completely positive requirement in the quantum evolution. We have sketched this reasoning in the scheme presented in figure 1.

\begin{figure}[t]\label{fig2}
\begin{center}
\includegraphics[width=0.8\textwidth]{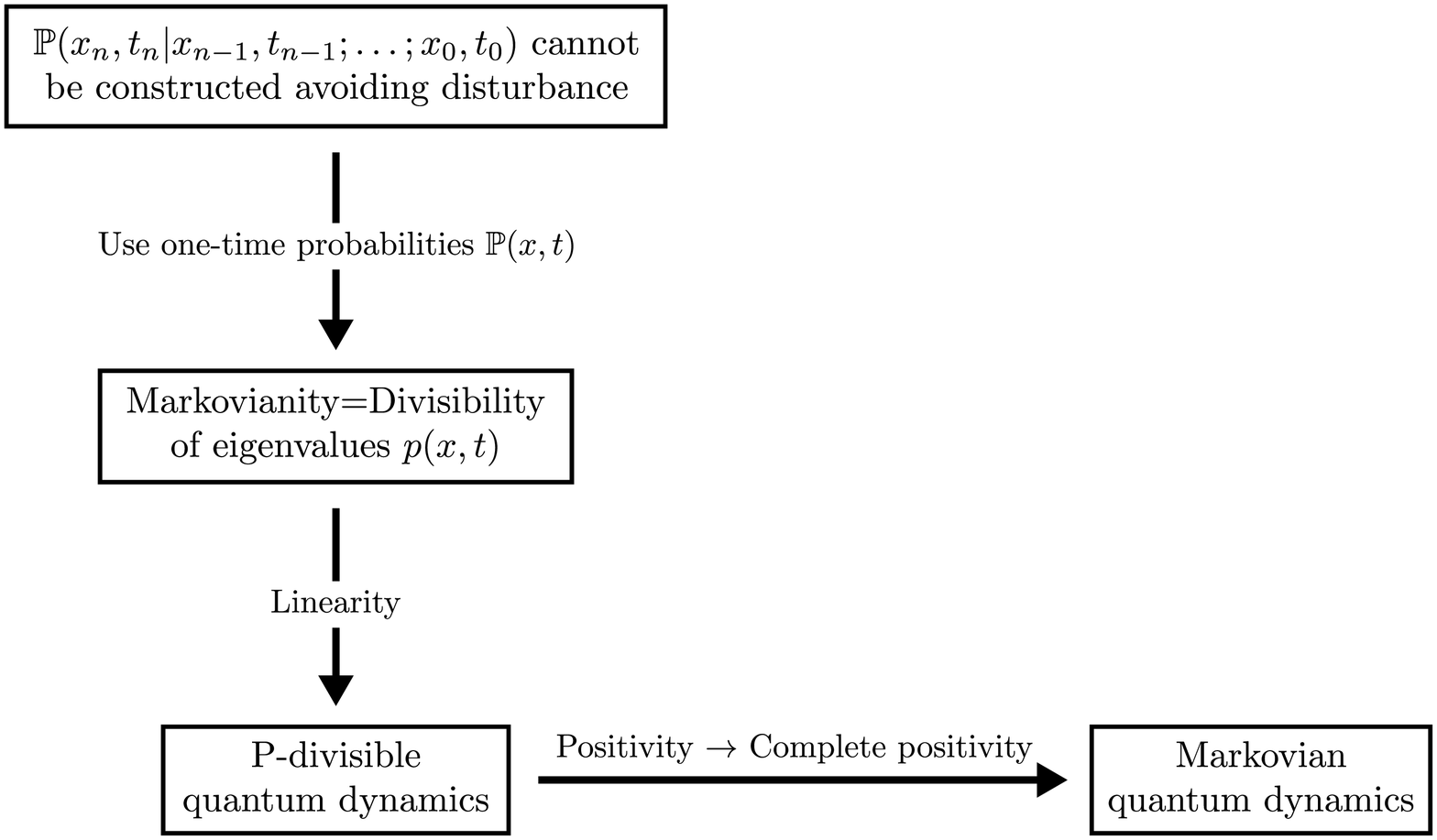}
\end{center}
\caption{Scheme of the arguments employed in the definition of quantum Markovian process (see main text). The equality in the second box is a consequence of Corollary \ref{Cor-Divisible}.}
\end{figure}

Finally, we review a fundamental result regarding differentiable quantum Markovian processes (i.e. processes such that the limit $\lim_{\epsilon\downarrow0}[\mathcal{E}_{(t+\epsilon,t)}-\mathds{1}]/\epsilon:=\mathcal{L}_t$ is well-defined). In this case, there is a mathematical result which is quite useful to characterize Markovian dynamics.

\begin{theorem} [Gorini-Kossakowski-Susarshan-Lindblad] An operator $\mathcal{L}_t$ is the generator of a quantum Markov (or divisible) process if and only if it can be written in the form \label{KossLindTheo}
\begin{equation}\label{diffMarkov}
\frac{d\rho(t)}{dt}=\mathcal{L}_t[\rho(t)] =-{\rm i}[H(t),\rho(t)]+\sum_k\gamma_k(t)\left[V_k(t)\rho(t)V_k^\dagger(t)-\frac{1}{2}\{V_k^\dagger(t)V_k(t),\rho(t)\}\right],
\end{equation}
where $H(t)$ and $V_k(t)$ are time-dependent operators, with $H(t)$ self-adjoint, and $\gamma_k(t)\geq0$ for every $k$ and time $t$.
\end{theorem}

This theorem is a consequence of the pioneering work by A. Kossakowski \cite{Kossakowski1,Kossakowski2} and co-workers \cite{GoKoSh76}, and independently G. Lindblad \cite{Lindblad76}, who analyzed the case of time-homogeneous equations, i.e. time-independent generators $\mathcal{L}_t\equiv\mathcal{L}$. For a complete proof including possible time-dependent $\mathcal{L}_t$ see \cite{Libro,TrazaNoHermitica}.

\subsection{Where is the memoryless property in quantum Markovian processes?}

As mentioned before, the motivation behind Definition  \ref{def-Marko} for quantum Markovian processes has been to keep a formal analogy with the classical case. However, it is not immediately apparent that the memoryless property present in the classical case is also present in the quantum domain. There are at least two ways to visualize this property which is hidden in Definition \ref{def-Marko}. As discussed below, one is based on the contractive properties of the completely positive maps and the other resorts to a collisional model of system-environment interactions.

\subsubsection{Contractive property of a quantum Markovian process}\label{sectionQcontraction}

Similarly to the classical case (see Section \ref{sec:contractiveClassical}), memoryless properties of quantum Markovian processes become quite clear in hypothesis testing problems \cite{BrLaPi1,PRAPolonia}. In the quantum case, we consider a system, with associated Hilbert space $\mathcal{H}$, whose state is represented by the density matrix $\rho_1$ with probability $q$, and $\rho_2$ with probability $(1-q)$. We wish to determine which density matrix describes the true state of the quantum system by performing a measurement. If we consider some general positive operator valued measure (POVM) $\{\Pi_x\}$ (cf. \cite{NC00}), where $x\in\mathcal{X}$ is the set of possible outcomes, we may split this set in two complementary subsets. If the outcome of the measurement is inside some $A\subset\mathcal{X}$, then we say that the state is $\rho_1$. Conversely, if the result of the measurement belongs to the complementary set $A^c$ such that $A\cup A^c=\mathcal{X}$, we say that the state is $\rho_2$. Let us group the results of this measurement in another POVM given by the pair $\{T,\mathbb{I}-T\}$, with $T=\sum_{x\in A} \Pi_x$.

Thus, when the true state is $\rho_1$ (which happens with probability $q$) we erroneously identify the state as $\rho_2$ with probability
\begin{eqnarray}
\sum_{j\in A^c} \Tr[\rho_1 \Pi_x]&=&\Tr\left[\rho_1 \left(\sum_{x\in A^c} \Pi_x\right)\right]=\Tr\left[\rho_1 (\mathbb{I}-T)\right].
\end{eqnarray}
On the other hand, when the true state is $\rho_2$ (which happens with probability $1-q$), we erroneously identify the state as $\rho_1$ with probability
\begin{equation}
\sum_{j\in A} \Tr[\rho_2 \Pi_x]=\Tr\left[\rho_2 \left(\sum_{x\in A} \Pi_x\right)\right]=\Tr\left[\rho_2 T\right].
\end{equation}
The problem in one-shot two-state discrimination is to examine the trade-off between the two error probabilities $\Tr\left[\rho_2 T\right]$ and $\Tr\left[\rho_1 (\mathbb{I}-T)\right]$. Thus, consider the best choice of $T$ that minimizes the total averaged error probability
\begin{align}
\min_{0\leq T\leq \mathbb{I}} \{(1-q)\Tr\left[\rho_2 T\right]+q\Tr\left[\rho_1 (\mathbb{I}-T)\right]\}
&=\min_{0\leq T\leq \mathbb{I}} \{q+\Tr\left[(1-q)\rho_2 T-q\rho_1 T\right]\}\nonumber\\
&=q-\max_{0\leq T\leq \mathbb{I}} [\Tr\left(\Delta T\right)],
\end{align}
where $\Delta=q\rho_1-(1-q)\rho_2$ is a Hermitian operator, with trace $\Tr \Delta=2q-1$ vanishing in the unbiased case $q=1/2$. $\Delta$ is sometimes called Helstrom matrix \cite{Helstrom}. We have the following result.

\begin{theorem} \label{BestMesurementT} With the best choice of $T$, the minimum total error probability in the one-shot two-state discrimination problem becomes
\begin{equation}\label{minimization}
\mathbb{P}_{\rm min}({\rm fail})=\min_{0\leq T\leq \mathbb{I}} \{(1-q)\Tr\left[\rho_2 T\right]+q\Tr\left[\rho_1 (\mathbb{I}-T)\right]\}=\frac{1-\|\Delta\|_1}{2},
\end{equation}
where $\|\Delta\|_1=\Tr\sqrt{\Delta^\dagger\Delta}$ is the trace norm of the Helstrom matrix $\Delta$.
\end{theorem}

Thus, note that when $q=0$ or $q=1$ we immediately obtain zero probability of wrongly identifying the true state.

\begin{proof} The proof follows the same steps as for the unbiased $q=1/2$ case (see \cite{NC00,Hayashi}). The spectral decomposition allows us to write $\Delta=\Delta^+-\Delta^-$, with positive operators $\Delta^\pm=\pm\sum_k\lambda_k^\pm\ket{\psi_k^{\pm}}\bra{\psi_k^{\pm}}$ where $\lambda_k^+$ are the positive eigenvalues of $\Delta$ and $\lambda_k^-$ the negative ones. Then it is clear that for $0\leq T\leq \mathbb{I}$
\begin{equation}
\Tr\left(\Delta T\right)=\Tr\left(\Delta^+T\right)-\Tr\left(\Delta^-T\right)\leq\Tr\left(\Delta^+T\right)\leq\Tr\left(\Delta^+\right),
\end{equation}
so that
\begin{equation}
\mathbb{P}_{\rm min}({\rm fail})=q-\max_{0\leq T\leq \mathbb{I}} [\Tr\left(\Delta T\right)]=q-\Tr\left(\Delta^+\right).\label{minimizationprueba}
\end{equation}
On the other hand, because $|\psi_j^\pm\rangle\langle\psi_j^\pm|$ are orthogonal projections (in other words as $\|\Delta\|_1=\sum_j|\lambda_j|$), the trace norm of $\Delta$ is
\begin{equation}
\|\Delta\|_1=\|\Delta^+\|_1+\|\Delta^-\|_1=\Tr(\Delta^+)+\Tr(\Delta^-).
\end{equation}
Since
\begin{equation}
\Tr(\Delta^+)-\Tr(\Delta^-)=\Tr(\Delta)=2q-1,
\end{equation}
we have
\begin{equation}
\|\Delta\|_1=2\Tr(\Delta^+)+(1-2q).
\end{equation}
Using this relation in \eqref{minimizationprueba} we straightforwardly obtain the result \eqref{minimization}.
\end{proof}

Thus the trace norm of $\Delta=q\rho_1-(1-q)\rho_2$ gives our capability to distinguish correctly between $\rho_1$ and $\rho_2$ in the one-shot two-state discrimination problem.

On the other hand, the following theorem connects trace-preserving and positive maps with the trace norm. It was first proven by Kossakowski in references \cite{Kossakowski1,Kossakowski2}, while Ruskai also analyzed the necessary condition in \cite{Ruskai}.

\begin{theorem}\label{theoCPTcontractions} A trace preserving linear map $\mathcal{E}$ is positive if and only if for any Hermitian operator $\Delta$ acting on $\mathcal{H}$,
\begin{equation}\label{contraction}
\|\mathcal{E}(\Delta)\|_1\leq\|\Delta\|_1.
\end{equation}
\end{theorem}
\begin{proof} Assume that $\mathcal{E}$ is positive and trace preserving, then for every positive semidefinite $\Delta\geq0$ the trace norm is also preserved, $\|\mathcal{E}(\Delta)\|_1=\|\Delta\|_1$. Consider $\Delta$ not to be necessarily positive semidefinite, then by using the same decomposition as in the proof of Theorem \ref{BestMesurementT}, $\Delta=\Delta^+-\Delta^-$, we have
\begin{align}
\|\mathcal{E}(\Delta)\|_1&=\|\mathcal{E}(\Delta^+)-\mathcal{E}(\Delta^-)\|_1 \nonumber \\
&\leq\|\mathcal{E}(\Delta^+)\|_1+\|\mathcal{E}(\Delta^-)\|_1=\|\Delta^+\|_1+\|\Delta^-\|_1=\|\Delta\|_1,
\end{align}
where the penultimate equality follows from the positivity of $\Delta^{\pm}$. Therefore, $\mathcal{E}$ fulfils Eq. \eqref{contraction}.

Conversely, assume that $\mathcal{E}$ satisfies Eq. \eqref{contraction} and preserves the trace, then for a positive semidefinite $\Delta$ we have the next chain of inequalities:
\[
\|\Delta\|_1=\Tr(\Delta)=\Tr[\mathcal{E}(\Delta)]\leq\|\mathcal{E}(\Delta)\|_1\leq\|\Delta\|_1, \quad \mathrm{for }\ \Delta\geq0,
\]
hence $\|\mathcal{E}(\Delta)\|_1=\mathrm{Tr}[\mathcal{E}(\Delta)]$. Since $\|\Delta\|_1=\mathrm{Tr}(\Delta)$ if and only if $\Delta\geq0$, we obtain that $\mathcal{E}(\Delta)\geq0$.

\end{proof}

There is a clear parallelism between this theorem and Theorem \ref{theoremClassi-contraction} for classical stochastic processes. As a result, quantum Markov processes are also characterized in the following way.

\begin{theorem} \label{contractionMarkov} A quantum evolution $\{\mathcal{E}_{(t_2,t_1)},t_2\geq t_1\geq t_0\}$ is Markovian if and only if for all $t_2$ and $t_1$, $t_2\geq t_1$,
\begin{equation}
\left\|\left[\mathcal{E}_{(t_2,t_1)}\otimes\mathds{1}\right](\tilde{\Delta})\right\|_1\leq\|\tilde{\Delta}\|_1,
\end{equation}
for any Hermitian operator $\tilde{\Delta}$ acting on $\mathcal{H}\otimes\mathcal{H}$.
\end{theorem}
\begin{proof} Since for a quantum Markovian process $\mathcal{E}_{(t_2,t_1)}$ is completely positive for any $t_2\geq t_1$, the map $\mathcal{E}_{(t_2,t_1)}\otimes\mathds{1}$ is positive, and the results follows from Theorem \ref{theoCPTcontractions}.
\end{proof}

\begin{figure}[t]
\begin{center}
\includegraphics[width=0.7\textwidth]{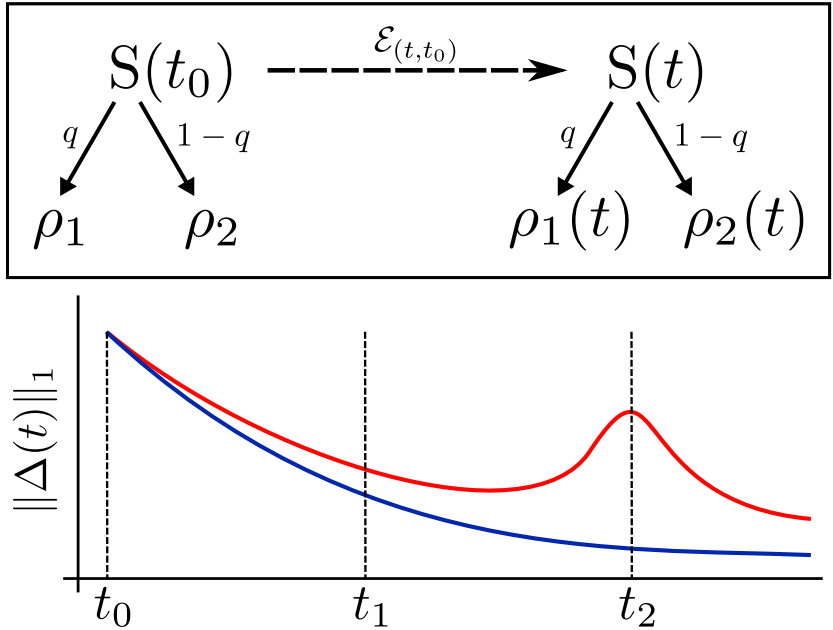}
\end{center}
\caption{Illustration of the quantum two-state discrimination problem. Under Markovian dynamics (blue line) the trace norm of the Helstrom matrix $\Delta(t)=q\rho_1(t)-(1-q)\rho_2(t)$ decreases monotonically from its initial value at $t_0$. On the contrary, for a non-Markovian dynamics (red line) there exist revivals at some time instants, say $t_2$, where the trace norm of $\Delta(t)$ is larger than the previously attained values, for example at $t_1$. Thus, the memoryless property of a Markovian process implies that information about the initial state is progressively lost as time goes by. That is not the case for non-Markovian dynamics.}
\end{figure}

Therefore, similarly to the classical case, a quantum Markovian process increases monotonically the averaged probability $\mathbb{P}_{\rm min}({\rm fail})$, Eq. \eqref{minimization}, to give the wrong answer in one-shot two-state discrimination problem. More concretely, consider a quantum system ``${\rm S}$'' which evolves from $t_0$ to the current time instant $t_1$, through some dynamical map $\mathcal{E}_{(t,t_0)}$. This system ${\rm S}$ was prepared at $t_0$ in the state $\rho_1$ with probability $q$ and $\rho_2$ with probability $(1-q)$, and we aim at guessing which state was prepared by performing a measurement on ${\rm S}$ at the present time $t_1$. If the dynamics $\mathcal{E}_{(t,t_0)}$ is Markovian the best we can do is to measure at the present time $t_1$, however for non-Markovian processes it may be better to wait for some posterior time $t_2>t_1$ where the trace norm of the Helstrom matrix $\Delta(t)=\mathcal{E}_{(t,t_0)}(\Delta)=q\rho_1(t)-(1-q)\rho_2(t)$ increases with respect to its value at time $t_1$ (see illustration in figure 2).

Moreover, the same result applies if we make measurements including a (and arbitrary dimensional) static ancillary system ``$\rm A$'', in such a way that the evolution of the enlarged system ``${\rm S+A}$'' is given by $\mathcal{E}_{(t_2,t_1)}\otimes\mathds{1}$. That is not the case for a P-divisible dynamics where just the positivity of $\mathcal{E}_{(t_2,t_1)}$ is required instead of the complete positivity.

The fact that for a quantum non-Markovian process $\mathbb{P}_{\rm min}({\rm fail})$ can decrease for some time period may again be interpreted as a signature of the underlying memory in the process. The system seems to retain information about its initial state, that arises at some posterior time $t_2$.

\subsubsection{Memoryless environment}\label{sec-collision}

A different way to visualize the memoryless properties characteristic of a quantum Markovian process is in the context of system-environment interactions. Since for a closed quantum system the evolution is given by some two-parameter family of unitary operators $U(t_1,t_0)$, which fulfill $U(t_2,t_0)=U(t_2,t_1)U(t_1,t_0)$ (see for instance \cite{GP90}), the evolution of a closed quantum system is trivially Markovian. However, the situation changes regarding the time evolution of open quantum systems. Despite the most studied models to describe such a dynamics result in Markovian master equations of the form of Eq. \eqref{diffMarkov} \cite{Davies,RevKoss,AlickiLendi87,GardinerZoller04}, it is well-known that the exact dynamics of an open quantum system is essentially non-Markovian \cite{Daviesbook76,BrPe02,Fain02,Libro}.

To illustrate in what sense Markovian dynamics are memoryless in this context, consider the collisional model in the formulation proposed in \cite{Ziman1,Ziman2,Ziman3}, which is depicted in figure 3. In this model, the interaction between system and environment is made up of a sequence of individual collisions at times $t_1, t_2,\ldots, t_n$. Each collision produces a change in the state of the system $\rho_{\rm S}$ given by
\begin{equation}
\rho_{\rm S}(t_{n+1})=\Tr_{\rm E} [U(t_{n+1},t_n)\rho_{\rm S}(t_{n})\otimes\rho_{\rm E}U^\dagger(t_{n+1},t_n)]=\mathcal{E}_{(t_{n+1},t_n)}[\rho_{\rm S}(t_{n})],
\end{equation}
where $\rho_{\rm E}$ is the state of the environment assumed to be the same for every collision, and $U(t_{n+1},t_n)$ is a unitary operator describing the system-environment  interaction. Moreover, $\mathcal{E}_{(t_{n+1},t_n)}(\cdot)=\sum_{ij} K_{ij}(\cdot)K_{ij}^\dagger$ is a completely positive map whose Kraus operators are given by $K_{ij}=\sqrt{p_j^{\rm E}}\langle\phi^i_{\rm E}|U(t_{n+1},t_n)|\phi^j_{\rm E}\rangle$, for $\rho_{\rm E}=\sum_j p_j^{\rm E}|\phi^j_{\rm E}\rangle\langle\phi^j_{\rm E}|$. The successive concatenations of these collisions lead to a quantum Markovian process. Indeed, if we write
\begin{equation}
\rho_{\rm S}(t_{n+2})=\mathcal{E}_{(t_{n+2},t_n)}[\rho_{\rm S}(t_{n})],
\end{equation}
as
\begin{equation}
\rho_{\rm S}(t_{n+2})=\mathcal{E}_{(t_{n+2},t_{n+1})}[\rho_{\rm S}(t_{n+1})]=\mathcal{E}_{(t_{n+2},t_{n+1})}\mathcal{E}_{(t_{n+1},t_{n})}[\rho_{\rm S}(t_{n})],
\end{equation}
we conclude that
\begin{equation}
\mathcal{E}_{(t_{n+2},t_n)}=\mathcal{E}_{(t_{n+2},t_{n+1})}\mathcal{E}_{(t_{n+1},t_{n})},
\end{equation}
and since $\mathcal{E}$ are completely positive maps the process is Markovian. In addition, if the limit $\max_n |t_{n+1}-t_n|\rightarrow0$ does exist, it is possible to obtain equations with the form of \eqref{diffMarkov} for these models \cite{Ziman2,Ziman3}.

\begin{figure}[t]
\begin{center}
\includegraphics[width=0.9\textwidth]{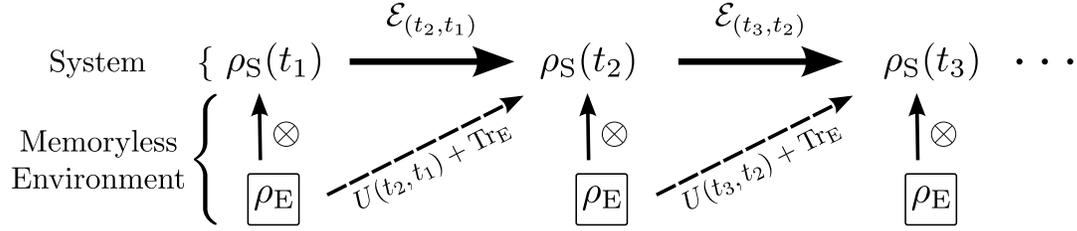}
\end{center}
\caption{Schematic action of a memoryless environment as described by a collisional model. At every time $t_n$ the system interacts with the environmental state $\rho_E$ via some unitary operation $U$. At the following time step $t_{n+1}$, the system finds again the environment in the same state $\rho_E$, forgetting any correlation caused of the previous interaction at $t_n$. Provided the limit $\max_{n}|t_{n+1}-t_{n}|\rightarrow0$ is well-defined, this process can be seen as a discrete version of quantum Markovian dynamics.}
\end{figure}

Notably, any Markovian dynamics can be seen as a collisional model like this. This is a consequence of the following theorem \cite{NC00,AlickiLendi87,Libro}.

\begin{theorem}[Stinespring \cite{Stinespring55}] \label{Stinespring} A completely positive dynamics $\mathcal{E}(\rho_{\rm S})$ can be seen as the reduced dynamics of some unitary evolution acting on an extended state with the form $\rho_{\rm S}\otimes\rho_{\rm E}$, where $\rho_{\rm E}$ is the same independently of $\rho_{\rm S}$.
\end{theorem}
Thus, since for Markovian evolutions $\mathcal{E}_{(t_2,t_1)}$ exists for all $t_2\geq t_1$ and is completely positive, we may write it as \begin{equation}
\mathcal{E}_{(t_2,t_1)}(\cdot)=\Tr_{\rm E} [U(t_2,t_1)(\cdot)\otimes\rho_{\rm E}U^\dagger(t_2,t_1)],
\end{equation}
where $U(t_2,t_1)$ may depend on $t_2$ and $t_1$, but $\rho_{\rm E}$ can be taken to be independent of time. Hence, Markovian evolutions may be thought to be made up of a sequence of memoryless collisions, where the environmental state is the same and the total state of system and environment is uncorrelated in every collision as if there were no previous interaction (figure 3). Note that this does not mean that we must impose system and environment to be uncorrelated at any time to get Markovian evolutions \cite{Libro}. Actually the total state of system and environment may be highly correlated even for dynamics leading to Markovian master equations \cite{Rivas09}. Rather, the conclusion is that the obtained evolution may also be thought as the result of memoryless system-environment infinitesimal collisions.

Interestingly, this kind of collisional models can be adapted to simulate non-Markovian dynamics by breaking the condition of uncorrelated collisions \cite{VirmaniPlenio1,VirmaniPlenio2,Ziman4,CollisionnoMarkov,Diosi}.

\subsection{Comparison with other definitions of quantum Markovianity}\label{sec:comparison}
Our approach which is based on the divisibility property is not the unique approach to non-Markovianity and indeed, alternative approaches are being pursued in the literature. Before moving on, it is therefore worthwhile to dedicate a brief section to present these alternative definitions of quantum Markovianity, to comment on these alternative lines of research and to refer the reader to the most relevant literature.

\begin{itemize}
\item Semigroup Definition. Historically, the absence of memory effects in quantum dynamics was commonly associated to the formulation of differential dynamical equations for $\rho(t)$ with time-independent coefficients. In contrast, differential equations with time-dependent coefficients or integro-differential equations were linked to non-Markovian dynamics (see for instance \cite{BrPe02,Shibata,Wilkie,Barnett,Royer,Budini,Daffer04,Lee,ShabaniLidar,ManiscalcoSola1,ManisPetru,Breuer-Vacchini,Koss-Rebo,ChruKosPas}, and references therein). From this point of view, Markovian evolutions would be given only by quantum dynamical semigroups \cite{AlickiLendi87}, i.e. families of trace preserving and completely positive maps, $\mathcal{E}_{\tau}$, fulfilling the condition
\begin{equation}\label{semigroupLaw}
\mathcal{E}_{\tau}\mathcal{E}_{\sigma}=\mathcal{E}_{\tau+\sigma}, \quad \tau,\sigma\geq 0.
\end{equation}
It should be noted however, that this definition does not coincide with the definition adopted in this review and, in our view, suffers from certain drawbacks. The semigroup law Eq. \eqref{semigroupLaw} is just a particular case of the two-parameter composition law Eq. \eqref{CompositionLaw}, which encompasses the case of time-inhomogeneous Markovian processes. In other words, this approach does not distinguish between non-Markovian and Markovian equations of motion with time-dependent coefficients. Moreover, this problem persists in the classical limit.

\item Algebraic Definition. In the 1980s a rigorous definition of quantum stochastic process was introduced by using the algebraic formulation of quantum mechanics \cite{AccFriLew,Lewis}. It is difficult to summarize in a few words those results, but we will try to sketch the main idea for those amongst the readership that are familiar with $C^{\ast}-$algebras. In this context, a quantum stochastic process on a $C^{\ast}-$algebra $\mathcal{A}$ with values in a $C^{\ast}-$algebra $\mathcal{B}$ is defined by a family $\{j_t\}_{t\geq0}$ of $\ast-$homomorphism $j_t:\mathcal{B}\rightarrow \mathcal{A}$. To define a Markov property two ingredients are necessary.  The first one is the following sub-algebra of $\mathcal{A}$,
\begin{equation}
\mathcal{A}_{t]}=\vee\{j_s(b):b\in\mathcal{B},s\leq t\}
\end{equation}
which is called a \emph{past filtration} or a \emph{filtration} \cite{Fagnola}. Here the symbol $\vee S$ denotes the  $C^{\ast}-$algebra generated by the subset $S$ of $\mathcal{A}$. The second one is the introduction of the concept of conditional expectation on $\mathcal{A}$ \cite{Umegaki,Takesaki,Petz,Fagnola}, which is a generalization of the usual conditional expectation, see Eq. \eqref{ConditionalMarkov}, to non-commutative algebras. Mathematically, a \emph{conditional expectation} of $\mathcal{A}$ on a sub-algebra $\mathcal{A}_0\subset\mathcal{A}$ is a linear map
\begin{equation}
\mathbb{E}[\ \cdot\ |\mathcal{A}_0]:\mathcal{A} \rightarrow \mathcal{A}_0,
\end{equation}
satisfying the properties:
\begin{enumerate}
\item For $a\in\mathcal{A}$, $\mathbb{E}[a|\mathcal{A}_0]\geq0$ whenever $a\geq0$.
\item $\mathbb{E}[\mathbb{I}|\mathcal{A}_0]=\mathbb{I}$.
\item For $a_0\in \mathcal{A}_0$ and $a\in \mathcal{A}$, $\mathbb{E}[a_0a|\mathcal{A}_0]=a_0\mathbb{E}[a|\mathcal{A}_0]$.
\item For $a\in\mathcal{A}$, $\mathbb{E}[a^\ast|\mathcal{A}_0]=(\mathbb{E}[a|\mathcal{A}_0])^\ast$.
\end{enumerate}
Thus, the stochastic process $\{j_t\}_{t\geq0}$ is said to be Markovian if for all $s, t\geq0$ and all $X\in\mathcal{A}_{0]}$ a condition analogous to Eq. \eqref{ConditionalMarkov} is fulfilled,
\begin{equation}\label{def-Marko-Alg}
\mathbb{E}[j_{t+s}(X)|\mathcal{A}_{s]}]=\mathbb{E}[j_{t+s}(X)|j_s(\mathcal{A}_{0]})].
\end{equation}

We will not go into further details here. What is important for our purposes is that, on the one hand,  Accardi, Frigerio and Lewis proved in their seminal paper \cite{AccFriLew} that this definition of Markovian process implies our Definition \ref{def-Marko} (rewritten in the Heisenberg picture). On the other hand, the opposite problem, namely to prove that any Markovian evolution according to Definition \ref{def-Marko} is also Markovian according to \eqref{def-Marko-Alg} requires a technically complicated step known as the \emph{dilation problem} (see \cite{dilation1} and references therein). That is quite closely related to what was explained informally in Section \ref{sec-collision}, but we do not enter into details here. Fortunately, under well-chosen and reasonable conditions (boundedness of operators, fulfilment of Lipschitz conditions, etc.) \cite{dilation1,dilation2,dilation3,dilation4,dilation5}, it is possible to prove that Definition \ref{def-Marko} also implies \eqref{def-Marko-Alg}. Therefore, within the scope this paper, i.e. finite dimensional systems, we can consider the algebraic definition of Markovianity to be essentially equivalent to the one given here in terms of the divisibility condition.

\item BLP Definition. Recently, Breuer, Laine and Piilo (BLP) proposed a definition of non-Markovian dynamics in terms of the behavior of the trace distance \cite{BrLaPi1,BrLaPi2,BreuerReview}. Concretely, they state that a quantum evolution, given by some dynamical map $\mathcal{E}_{(t,t_0)}$, is Markovian if the trace distance between any two initial states $\rho_1$ and $\rho_2$ decreases monotonically with time. This definition is a particular case of Definition \ref{def-Marko}. As was explained in Section \ref{sectionQcontraction}, for any Markovian dynamics $\mathcal{E}_{(t,t_0)}$ and Hermitian operator $\tilde{\Delta}$ acting on $\mathcal{H}\otimes\mathcal{H}$, $\|\left[\mathcal{E}_{(t,t_0)}\otimes\mathds{1}\right](\tilde{\Delta})\|_1$ monotonically decreases with time, and so does $\|\mathcal{E}_{(t,t_0)}(\Delta)\|_1$ for any Hermitian $\Delta$ acting on $\mathcal{H}$. Concretely, for $\Delta=\tfrac{1}{2}(\rho_1-\rho_2)$, which corresponds to the unbiased case in the two-state discrimination problem $q=1/2$, the property
\begin{equation}
\|\mathcal{E}_{(t_2,t_0)}(\Delta)\|_1\leq\|\mathcal{E}_{(t_1,t_0)}(\Delta)\|_1, \quad t_2\geq t_1,
\end{equation}
reduces to the BLP definition; this is, for all $\rho_1$ and $\rho_2$,
\begin{equation}\label{BPLTrace}
\|\rho_1(t_2)-\rho_2(t_2)\|_1\leq\|\rho_1(t_1)-\rho_2(t_1)\|_1, \quad t_2\geq t_1.
\end{equation}
However, the reverse implication fails to hold, i.e. not every dynamics fulfilling Eq. \eqref{BPLTrace} satisfies Theorem \ref{contractionMarkov} (e.g. \cite{Br-Q-Bos5,PRAPolonia,DarekFilip}). Thus, we believe that it is more appropriate to consider the BLP definition as a particular case which arises in the study of memory properties in unbiased two-state discrimination problems. Note that the apparent lack memory in an unbiased case does not imply a general memoryless property; it only manifested in a general biased case (and taking into account possible ancillary systems). Nevertheless, from equation \eqref{BPLTrace}, it is possible to construct a very useful witness of non-Markovianity as we will see in Section \ref{secTracedistance}.
\end{itemize}
Remarkably, the previous different definitions of quantum Markovianity satisfy a hierarchical relation with our Definition \ref{def-Marko} based on the divisibility condition. That is sketched in figure 4.
\begin{figure}[h]
\begin{center}
\includegraphics[width=0.6\textwidth]{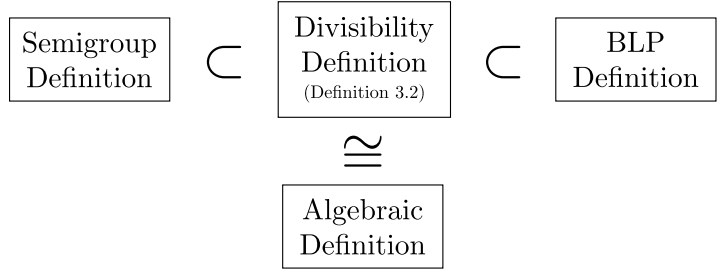}\label{fig4}
\end{center}
\caption{Relation between divisibility, semigroup, algebraic and BLP definitions of quantum Markovianity. The divisibility definition is essentially equivalent to the algebraic one (see main text). In addition, any Markovian dynamics according to the semigroup definition, it is also Markovian according to the divisibility definition, and hence Markovian according to BLP definition. However the converse implication does not hold.}
\end{figure}
\begin{itemize}
\item Markovianity in microscopic derivations. When deriving evolution equations for open quantum systems from microscopic models, the adjective ``Markovian'' is widely used to design master equations obtained under the so-called ``Born-Markov'' approximation. More concretely, if $\rho(t)$ is the state of the open quantum system, the Born approximation truncates the perturbative expansion in the interaction Hamiltonian, $V=\sum_i A_i\otimes B_i$, at first non-trivial order. This leads to some differential equation of the form \cite{vanKampen07,BrPe02,Libro,Weiss08,GardinerZoller04,Fain02,WisemanMilburn,Whitney08}:
    \begin{equation}\label{Bloch-Redfield}
    \frac{d\rho(t)}{dt}=-{\rm i}[H_S,\rho(t)]+\sum_{i,j} \Omega_{ij}(t)\rho(t)A_j+A_j\rho(t)\Omega_{ij}^\dagger(t)-A_j\Omega_{ij}(t)\rho(t)-\rho(t)\Omega_{ij}^\dagger(t)A_j.
    \end{equation}
    where $H_S$ stands for the free Hamiltonian of the open system, and
    \begin{equation}
    \Omega_{ij}(t)=\int_0^t ds C_{ij}(s){\rm e}^{-{\rm i}H_Ss}A_i{\rm e}^{{\rm i}H_Ss}.
    \end{equation}
    Here, $C_{ij}(s)=\Tr(B_j{\rm e}^{-{\rm i}H_Bs}B_i{\rm e}^{{\rm i}H_Bs}\rho_B)$ are the correlation functions of the bath, which is in the state $\rho_B$ and has free Hamiltonian $H_B$. Eq. \eqref{Bloch-Redfield} is sometimes called Bloch-Redfield equation (e.g. \cite{Whitney08}). Now, if the correlation functions of the bath $C_{ij}(s)$ are narrow in comparison to the typical time scale of $\rho(t)$ due to $V$, the upper limit in the integral of $\Omega_{ij}(t)$ can be safely extended to infinity. This conforms what is sometimes called ``Markov'' approximation.

    Two comments are pertinent regarding the connection of these dynamical equations with the Markovian processes as defined this work. Firstly, despite the fact that the ``Born-Markov'' approximation leads to master equations with time-independent coefficients, they do not always define a valid quantum dynamical semigroup \cite{DumckeandSpohn,ZhaoChen02}. This is because they break complete positivity. Thus, these models should not be referred as ``Markovian'' in strict sense, as a Markovian processes must preserve the positivity of any state, or any probability distribution in the classical limit. Secondly, if the ``Born-Markov'' approximation is combined with the secular approximation (i.e. neglecting fast oscillating terms in the evolution equation) a valid quantum dynamical semigroup is obtained \cite{Davies,RevKoss,AlickiLendi87,BrPe02,Libro}, and then the dynamics can be certainly called Markovian. However, the fact that the ``Born-Markov-secular'' approximation generates Markovian dynamics, should not be understood as the only framework to obtain Markovian dynamics.

\end{itemize}

Further to this short summary of definitions for quantum non-Markovianity different from Definition \ref{def-Marko}, the reader may also find proposals based on the behavior of multi-time correlation functions \cite{LindbladDefinition,Ines,Petruccione}, initial-time-dependent generators \cite{t0definition1,t0definition2,t0definition3}, or properties of the asymptotic state \cite{ChruKosPas}. See also \cite{Diosi1,Diosi2,Diosi3,Gaspard} for a definition of non-Markovianity in the context of stochastic Schr\"odinger equations.

\section{Measures of Quantum non-Markovianity} \label{section:measures}

After introducing the concept of quantum non-Markovianity in previous sections, we may ask about its quantification in terms of suitable measures and its detection in actual experiments. As we shall see, recently there have been several developments towards these goals, and we shall present them separately. Thus, the present section is devoted to the quantification problem whereas the detection of non-Markovian dynamics by witnesses is left to Section \ref{section:witness}.

In order to quantify non-Markovianity, the so-called \emph{measures of non-Markovianity} are introduced. Basically, a measure of non-Markovianity is a function which assigns a number (positive or zero) to each dynamics, in such a way that the zero value is obtained if and only if the dynamics is Markovian. We will also use the name \emph{degree of non-Markovianity} for a normalized measure of non-Markovianity, with values between 0 and 1, although other normalizations may eventually be taken.

\subsection{Geometric measures}
Consider a dynamical map $\mathcal{E}_{(t,t_0)}$ describing the evolution from some initial time $t_0$. A first attempt to formulate a measure of non-Markovianity may be a distance-based approach. Here the measure of non-Markovianity is expressed as a distance between $\mathcal{E}_{(t,t_0)}$ and  its closest Markovian dynamics (see figure 5). Specifically, let $\mathfrak{M}$ denote the set of all Markovian dynamics, and $\mathcal{D}(\mathcal{E}_1,\mathcal{E}_2)\in[0,1]$ be some (normalized) distance measure in the space of dynamical maps. We define the \emph{geometric non-Markovianity} at time $t$ as
\begin{equation}
\mathcal{N}^{\rm geo}_t[\mathcal{E}_{(t,t_0)}]:=\min_{\mathcal{E}^M\in\mathfrak{M}}\mathcal{D}[\mathcal{E}_{(t,t_0)},\mathcal{E}^M_{(t,t_0)}],
\end{equation}
which is zero if and only if $\mathcal{E}_{(t,t_0)}$ belongs to the set of Markovian dynamics $\mathfrak{M}$.

The geometric measure of non-Markovianity in some time interval $I$ may be defined as the maximum value of the geometric non-Markovianity for $t\in I$,
\begin{equation}\label{geoDegree}
\mathscr{D}_{\rm NM (g)}^I:=\max_{t\in I}\mathcal{N}^{\rm geo}_t[\mathcal{E}_{(t,t_0)}].
\end{equation}
This quantity lies between 0 and 1 and is positive if an only if the process is non-Markovian, therefore it is a degree of non-Markovianity.

\begin{figure}
\begin{center}
\includegraphics[width=0.7\textwidth]{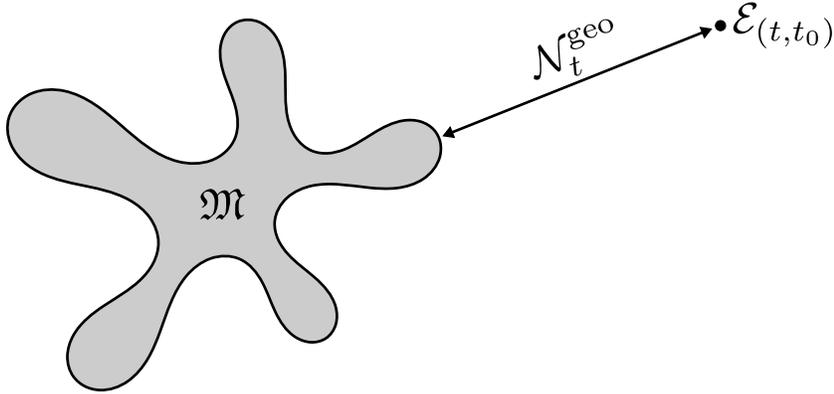}
\end{center}
\caption{Illustration of the geometric measure of non-Markovianity. At each $t$, $\mathcal{N}^{\rm geo}_t[\mathcal{E}_{(t,t_0)}]$ measures the distance between the map $\mathcal{E}_{(t,t_0)}$ and the non-convex set of Markovian maps $\mathfrak{M}$. For a time interval $t\in I$, $\mathscr{D}_{\rm NM (g)}^I$ in Eq. \eqref{geoDegree} is the maximum of every value of $\mathcal{N}^{\rm geo}_t[\mathcal{E}_{(t,t_0)}]$ for $t\in I$.}
\end{figure}

Despite the conceptually clear meaning of $\mathscr{D}_{\rm NM (g)}^I$, it suffers from an important drawback, as it is very hard to compute in practice because of the involved optimization process. In fact, note that the set of Markovian maps $\mathfrak{M}$ is non-convex \cite{Wolf2}, which makes the problem computationally intractable as the dimension of the system grows.

This approach was originally proposed by Wolf and collaborators \cite{Wolf1,Wolf2} to quantify the non-Markovianity of a quantum channel. A quantum channel is a completely positive and trace preserving map $\mathcal{R}$ acting on the set of quantum states. Then, $\mathcal{R}$ is said to be Markovian if it is the ``snapshot of some Markovian dynamics'', i.e. there exists some Markovian dynamics specified by the family of maps $\{\mathcal{E}_{(t,t_0)}, t\geq t_0\}$ such that $\mathcal{R}=\mathcal{E}_{(t_1,t_0)}$ for some $t_1\geq t_0$. Those authors put forward the aforementioned practical problems of the geometric measure of non-Markovianity and introduced an alternative measure (see also \cite{DavidWolf}). Later on, the problem to decide whether a quantum channel is Markovian was shown to be very hard in complexity theory sense \cite{Wolf3,Wolf4}, however it has been analyzed for small-size systems in \cite{BrLaPi2} and \cite{Helm}. For a review about quantum channels with memory see \cite{FilippoReview}.

\subsection{Optimization of the Helstrom matrix norm}\label{section:OptimizationHelstrom}

Another approach to quantify non-Markovianity is based on the result of Theorem \ref{contractionMarkov}. Recall that the trace norm of a Helstrom matrix $\Delta=q\rho_1-(1-q)\rho_2$ is a measure of the capability to distinguish between the states $\rho_1$ and $\rho_2$ given the outcome of some POVM, see Theorem \ref{BestMesurementT}. Thus, if the dynamics is such that for some $t$ and $\epsilon>0$, $\|\Delta(t)\|_1<\|\Delta(t+\epsilon)\|_1$, the probability to distinguish whether the system was in state $\rho_1$ or $\rho_2$ at time $t_0$, is higher at $t+\epsilon$ than it was a time $t$. As commented in Section \ref{sectionQcontraction}, this phenomenon denotes the existence of memory effects in the dynamics, as an increase of information at time $t+\epsilon$ with respect to $t$ suggests that the system is ``remembering'' its original state at $t+\epsilon$. In fact, the intuitive understanding of the word ``memory'' demands that a memoryless process does not have the property to keep information, and that this always decreases with time.

This revival of information at $t+\epsilon$ may be understood as a positive flow of information from the environment to the system. Thus, for purely Markovian dynamics the flow of information goes always from the system to the environment. However, as pointed out in \cite{Enta-3,Pernice,SaroReview}, this interpretation in terms of information flowing between system and environment may be problematic if taken strictly, because it is possible to obtain quantum non-Markovian dynamics with the form
\begin{equation}
\mathcal{E}_{(t,t_0)}(\rho)=\sum_i p_iU_i(t,t_0)\rho U_i^\dagger(t,t_0).
\end{equation}
This type of evolution can be generated simply by applying randomly the unitary evolutions $U_i(t,t_0)$ in accordance to the probabilities $p_i$. It is a fact that these probabilities can be generated independently of the dynamics of $\rho$ by some random (or pseudo-random) number generator.

On the other hand, since we cannot discard the presence of a decoupled and inert, arbitrary dimensional, ancillary space ``A'' (actually, it is enough to consider $\dim\mathcal{H}_{\rm A}=\dim\mathcal{H}$), we generally take an enlarged Helstrom matrix $\tilde{\Delta}=q\rho_{1{\rm A}}-(1-q)\rho_{2{\rm A}}$, where $\Delta=\Tr_{\rm A}(\tilde{\Delta})$, and $\|\tilde{\Delta}(t)\|_1=\|[\mathcal{E}_{(t,t_0)}\otimes\mathds{1}][\tilde{\Delta}(t_0)]\|_1$. Thus, as an increment of information as accounted for by $\|\tilde{\Delta}\|_1$ denotes non-Markovianity in the dynamics, we can take the maximum of information gained to assess how non-Markovian the evolution is. Explicitly, we may write
\begin{equation}\label{tildesigma}
\tilde{\sigma}(\tilde{\Delta},t):=\frac{d\|\tilde{\Delta}(t)\|_1}{dt}:=\lim_{\epsilon\rightarrow0^+}\frac{\|\tilde{\Delta}(t+\epsilon)\|_1-\|\tilde{\Delta}(t)\|_1}{\epsilon},
\end{equation}
by adding up every increment of information in some interval $I$:
\begin{equation}\label{tildesigmaInt}
\int_{t\in I,\tilde{\sigma}>0}dt\tilde{\sigma}(\tilde{\Delta},t).
\end{equation}
Then by maximizing over the initial Helstrom matrix $\tilde{\Delta}$ (i.e. maximizing over $\rho_{1{\rm A}}$, $\rho_{2{\rm A}}$ and the bias $q$) we define
\begin{equation}\label{NHelstrom}
\mathcal{N}_{\rm H}^I:=\max_{\tilde{\Delta}} \int_{t\in I,\tilde{\sigma}>0}dt\tilde{\sigma}(\tilde{\Delta},t),
\end{equation}
where the subindex ``H'' stands for Helstrom, as a measure of non-Markovianity. In virtue of Theorem \ref{contractionMarkov}, $\mathcal{N}_{\rm H}^I=0$ if and only if the process is Markovian in the interval $I$. The quantity $\mathcal{N}_{\rm H}^I$ can be normalized via exponential or rational functions, for instance $\mathscr{D}_{\rm NM (exp-H)}^I:=1-{\rm e}^{-\mathcal{N}_{\rm H}^I}$ or $\mathscr{D}_{\rm NM (rat-H)}^I:=\mathcal{N}_{\rm H}^I/(1+\mathcal{N}_{\rm H}^I)$.

This proposal was first suggested in \cite{PRAPolonia}. For the unbiased case $q=1/2$ and without taking into account the possible presence of ancillary systems it was previously formulated in \cite{BrLaPi1}. As in the case of geometric measures, the main drawback of the quantity $\mathcal{N}_{\rm H}^I$ is the difficult optimization process, which makes this measure rather impractical. For the restricted case of \cite{BrLaPi1} there has been some progress along this line \cite{BreuerMax,BreuerMax2yExperimento}, see also Section \ref{secTracedistance} and references therein.

\subsection{The RHP measure}

As pointed out, even though the two previous measures enjoy several nice geometric or informational interpretations, they are very difficult to compute in practice. A computationally simpler measure of non-Markovianity was introduced in \cite{nuestro} by Rivas, Huelga and Plenio. Given a family $\{\mathcal{E}_{(t,t_0)}, t\geq t_0\}$, the basic idea of this measure is to quantify how much non-completely positive the intermediate dynamics $\{\mathcal{E}_{(t,t_1)}, t\geq t_1\geq t_0\}$ is for every time $t_1$. To obtain these partitions, by time-continuity we have

\begin{equation}
\mathcal{E}_{(t,t_0)}=\mathcal{E}_{(t,t_1)}\mathcal{E}_{(t_1,t_0)}.
\end{equation}
After right-multiplication with the inverse of $\mathcal{E}_{(t_1,t_0)}$ on both sides we obtain the desired partitions
\begin{equation}\label{intermediateE}
\mathcal{E}_{(t,t_1)}=\mathcal{E}_{(t,t_0)}\mathcal{E}^{-1}_{(t_1,t_0)}.
\end{equation}
If these maps are completely positive (CP) for all $t_1$, the time evolution is Markovian (Definition \ref{def-Marko}). At the moment, we shall assume that $\mathcal{E}^{-1}_{(t_1,t_0)}$ does exist (we will come back to this point later on), so that $\mathcal{E}_{(t,t_1)}$ is well defined by Eq. \eqref{intermediateE}. For non-Markovian dynamics there must be some $t_1$, such that $\mathcal{E}_{(t,t_1)}$ is not completely positive. Therefore by measuring how much the intermediate dynamics $\{\mathcal{E}_{(t,t_1)}, t\geq t_1\geq t_0\}$ depart from completely positive maps, we are measuring up to what extent the time evolution is non-Markovian. Note that $\mathcal{E}_{(t,t_1)}$ is always trace-preserving as it is a composition of two trace-preserving maps.

In order to quantify the degree of non-complete positiveness of the maps $\{\mathcal{E}_{(t,t_1)}, t\geq t_1\geq t_0\}$, we resort to the Choi-Jamio{\l}kowski isomorphism \cite{Choi,Jamiolkowski}. Consider the maximally entangled state between two copies of our system $|\Phi\rangle=\frac{1}{\sqrt{d}}\sum_{n=0}^{d-1}|n\rangle |n\rangle$ (here $d$ denotes the dimension), we associate the map $\mathcal{E}_{(t,t_1)}$ to a (Choi-Jamio{\l}kowski) matrix constructed by the rule
\begin{equation}\label{ChoiMatrix}
\left[\mathcal{E}_{(t,t_1)}\otimes\mathds{1}\right]\left(|\Phi\rangle\langle\Phi|\right).
\end{equation}
The Choi's theorem asserts that $\mathcal{E}_{(t,t_1)}$ is completely positive if and only if the matrix Eq. \eqref{ChoiMatrix} is positive semidefinite. In addition, note that since $\mathcal{E}_{(t,t_1)}$ is trace-preserving, the trace norm of matrix \eqref{ChoiMatrix} provides a measure of the non-completely positive character of $\mathcal{E}_{(t,t_1)}$. More concretely,
\begin{equation}
\left\|\left[\mathcal{E}_{(t,t_1)}\otimes \mathds{1}\right]\left(|\Phi\rangle\langle\Phi|\right)\right\|_1 \left\{\begin{array}{ll}
=1 & \mbox{ iff } \mathcal{E}_{(t,t_1)} \mbox{ is CP},\\
>1 & \mbox{ otherwise.}\end{array}\right.
\end{equation}
Following \cite{nuestro}, we define a function $g(t)$ via the right derivative of the trace norm as
\begin{equation}\label{gfunction}
g(t):=\lim_{\epsilon\rightarrow0^+}\frac{\left\Vert\left[\mathcal{E}_{(t+\epsilon,t)}\otimes\mathds{1}\right]\left(|\Phi\rangle\langle\Phi|\right)\right\Vert_1-1}{\epsilon},
\end{equation}
so that $g(t)>0$ for some $t$ if and only if the evolution is non-Markovian. Therefore the total amount of non-Markovianity in an interval $t\in I$ will be given by
\begin{equation}\label{CJI}
\mathcal{N}_{\rm RHP}^I:=\int_I g(t)dt,
\end{equation}
where ``RHP'' stands for Rivas, Huelga and Plenio \cite{nuestro}.

The quantity $\mathcal{N}_{\rm RHP}^I$ may be normalized via exponential or rational methods, for instance $\mathscr{D}_{\rm NM (exp-RHP)}^I:=1-{\rm e}^{-\mathcal{N}_{\rm RHP}^I}$ or $\mathscr{D}_{\rm NM ( rat-RHP)}^I:=\mathcal{N}_{\rm RHP}^I/(1+\mathcal{N}_{\rm RHP}^I)$. However these normalizations turn out to be not very discriminative. Another way to obtain a more useful normalized measure was proposed in \cite{normalizationChinos}, we explain here a slightly modified method in terms of $g(t)$.

Define the function
\begin{equation}\label{gbarfunction}
\bar{g}(t):=\tanh[g(t)],
\end{equation}
where $g(t)$ is given by Eq. \eqref{gfunction}. Therefore $1\geq\bar{g}(t)\geq0$ with $\bar{g}(t)=0$ for all $t$ if and only if the evolution is Markovian. Then, for a bounded interval $t\in I$ we define its normalized degree of non-Markovianity as
\begin{equation}\label{CJdegree}
\mathscr{D}_{\rm NM (RHP)}^I:=\frac{\int_I \bar{g}(t)dt}{\int_I \chi[\bar{g}(t)]dt}, \quad \text{(with ``$0/0=0$'' by convention),}
\end{equation}
where the indicator function $\chi(x)$ is defined as
\begin{equation}
\chi(x):=\left\{\begin{array}{ll}
0 & \mbox{ if }x=0,\\
1 & \mbox{ otherwise.}\end{array}\right.
\end{equation}
Thus, the degree Eq. \eqref{CJdegree} is basically the non-Markovianity accumulated for each $t\in I$, divided by the total length of the subintervals of $I$ where the dynamics is non-Markovian. It is easy to proof that  $\mathscr{D}_{\rm NM (RHP)}^I$ is normalized. Let $I_n\subset I$ be the collection of subintervals such that $\bar{g}(t)>0$ for $t\in I_n$. If $|I_n|$ denotes the length of the subinterval $I_n$, we have
\begin{equation}
\mathscr{D}_{\rm NM (RHP)}^I=\frac{\int_I \bar{g}(t)dt}{\int_I \chi[\bar{g}(t)]dt}=\frac{\sum_n\int_{I_n} \bar{g}(t)dt}{\sum_n\int_{I_n} \chi[\bar{g}(t)]dt}=\frac{\sum_n\int_{I_n} \bar{g}(t)dt}{\sum_n |I_n|}\leq\frac{\sum_n |I_n| }{\sum_n |I_n|}=1,
\end{equation}
because of the bound $\bar{g}(t)\leq1$.

It is worth mentioning several points that one should keep in mind when evaluating this measure of non-Markovianity.
\begin{enumerate}
\item Note that, in general [see an exception in the point (iii) below], we need to know the complete dynamical map $\{\mathcal{E}_{(t,t_0)}, t\geq t_0\}$ to compute the function $g(t)$. The standard way to obtain it is resorting to process tomography. Thus, one considers the evolution for different final times $t$ of a complete set of states which span the space of dynamical maps. Then the dynamical map is reconstructed by tomography of the evolved final states (see for instance \cite{NC00}). This is the only experimental way to proceed. However if we know the theoretical evolution, for example by mean of some model, there is a trick which sometimes helps. In that case, we may consider directly the evolution of the basis $\{|i\rangle\langle j|\}$ for different final times $t$. We write the resulting matrix as a (column) vector by stacking the columns on top of one another. This process is sometimes called \emph{vectorization} and denoted by ``${\rm vec}$'' \cite{vec1,vec2}. As a result, the dynamical map $\mathcal{E}_{(t,t_0)}$ can be seen as a matrix $\mathbf{E}_{(t,t_0)}$ acting on states written as (column) vectors, and moreover in the basis of $\{|i\rangle\langle j|\}$ such a matrix is given by
    \begin{equation}
    \mathbf{E}_{(t,t_0)}=\big[\bm{v}_{11}(t),\ldots,\bm{v}_{1N}(t),\bm{v}_{21}(t),\ldots,\bm{v}_{2N}(t),\ldots\ldots,\bm{v}_{NN}(t)\big]
    \end{equation}
    where the column vectors are $\bm{v}_{ij}(t)={\rm vec}[|i\rangle\langle j|(t)]$. These are the vectorization of the matrix $|i\rangle\langle j|(t)$, which denotes the matrix obtained by evolving the basis element $|i\rangle\langle j|$ from $t_0$ to $t$.

    Once $\mathbf{E}_{(t,t_0)}$ is known for some interval $t\in I$, we can compute the intermediate dynamics in $I$ accordingly to Eq. \eqref{intermediateE}, $\mathbf{E}_{(t,t_1)}=\mathbf{E}_{(t,t_0)}\mathbf{E}^{-1}_{(t_1,t_0)}$, where $\mathbf{E}^{-1}_{(t_1,t_0)}$ is just the standard matrix inverse. Finally, $g(t)$ can be computed in the following way: first, construct the matrix $U_{2\leftrightarrow 3}[\mathbf{E}_{(t+\epsilon,t)}\otimes\mathbb{I}]U_{2\leftrightarrow 3}$ where $U_{2\leftrightarrow 3}$ is the commutation (or ``swap'') matrix between the ``second'' and the ``third'' subspace \cite{footnoteCommMatrix}; second, apply $U_{2\leftrightarrow 3}[\mathbf{E}_{(t+\epsilon,t)}\otimes\mathbb{I}]U_{2\leftrightarrow 3}$ on ${\rm vec}(|\Phi\rangle\langle\Phi|)$; third, write the result as a matrix, i.e. ``devectorize''; forth, compute the trace norm of that matrix which will correspond to $\left\Vert\left[\mathcal{E}_{(t+\epsilon,t)}\otimes\mathds{1}\right]\left(|\Phi\rangle\langle\Phi|\right)\right\Vert_1$; and finally, evaluate the right limit of Eq. \eqref{gfunction}.

\item It may happen that for some $t_1$ the map $\mathcal{E}_{(t_1,t_0)}$ is not bijective, so that the intermediate map $\mathcal{E}_{(t,t_1)}$ given by Eq. \eqref{intermediateE} is ill-defined. There are several ways to deal with this problem. If the singularity in $t_1$ is isolated, and we know the dynamics in some neighborhood of $t_1$, one can evaluate the function $g(t)$ in this neighboring region of $t_1$. By taking the limit $t\rightarrow t_1$ we usually obtain a divergence, $\lim_{t\rightarrow t_1} g(t)\rightarrow\infty$. However, since the hyperbolic tangent removes the divergence $\lim_{t\rightarrow t_1} \bar{g}(t)=1$, we can compute $\mathscr{D}_{\rm NM (RHP)}^I$ without further problems. Another way to remove the singularity may be to compute indirectly the inverse of $\mathcal{E}_{(t_1,t_0)}$ by finding the inverse of $\mathds{1}\epsilon+\mathcal{E}_{(t_1,t_0)}$, which always exist. Then, at the end of the computation of $\bar{g}(t)$, we proceed by taking the limit $\epsilon\rightarrow0$. Another more sophisticated (and in sometimes inequivalent) method has been proposed by using the Moore-Penrose pseudoinverse \cite{pseudoinverses}. See also \cite{ChinosSingulares,DMaldonado} for other considerations about singularities in dynamical maps.
\item There are cases where we know the dynamics fulfills some linear differential equation,
\begin{eqnarray}
\frac{d\rho(t)}{dt}=\mathcal{L}_t\left[\rho(t)\right]=&-&{\rm i}[H(t),\rho(t)]\\
&+&\sum_k\gamma_k(t)\left[V_k(t)\rho(t)V_k^\dagger(t)-\frac{1}{2}\{V_k^\dagger(t)V_k(t),\rho(t)\}\right]\nonumber,
\end{eqnarray}
where the decay rates may be negative $\gamma_k(t)<0$ for some $t$ and so describes non-Markovian evolutions. Then, there is a very practical way to obtain the function $g(t)$. Since for small enough $\epsilon$ we have \cite{Libro}
\begin{equation}
\mathcal{E}_{(t+\epsilon,t)}=\mathcal{T}\exp\left(\int_t^{t+\epsilon}\mathcal{L}_sds\right)\simeq\exp\left(\mathcal{L}_t \epsilon\right)\simeq\mathds{1}+\epsilon\mathcal{L}_t,
\end{equation}
the function $g(t)$ can be computed directly from the generator $\mathcal{L}_t$:
\begin{equation}\label{gfunctionL}
g(t)=\lim_{\epsilon\rightarrow0^+}\frac{\left\Vert\left[\mathds{1}+\epsilon(\mathcal{L}_t\otimes\mathds{1})\right](|\Phi\rangle\langle\Phi|)\right\Vert_1-1}{\epsilon}.
\end{equation}
\item It is possible to extend the definition Eq. \eqref{CJdegree} to unbounded intervals, typically $I=[t_0,\infty)$. However this extension must be carefully handled. It can be understood as a limiting procedure of bounded intervals $I_n$, such that $\lim_{n\rightarrow\infty}{I_n}=[t_0,\infty)$, for example $I_n=[t_0,n)$. Very crucially this limit has to be taken at the last step in the computation:
    \begin{equation}
    \mathscr{D}_{\rm NM (RHP)}^{[t_0,\infty)}:=\lim_{n\rightarrow\infty}\mathscr{D}_{\rm NM (RHP)}^{I_n}, \quad \mbox{with }\lim_{n\rightarrow\infty}{I_n}=[t_0,\infty).
    \end{equation}
\end{enumerate}

\begin{ex} Consider the following dynamical map of a two-dimensional quantum system (qubit), describing the evolution from $t_0=0$ (without loss of generality), \begin{equation}
\mathcal{E}_{(t,0)}(\rho)=[1-p(t)]\rho+p(t)\sigma_z\rho\sigma_z,\quad \text{where } p(t)\in[0,1],
\end{equation}
and $\sigma_z$ is the Pauli matrix. This dynamics describes the process where the nondiagonal elements (coherences) of $\rho$ change the sign with probability $p(t)$, and with probability $1-p(t)$ the qubit remains in the same state $\rho$. Note that for $p(t)=1/2$, the coherences vanish completely.  Let us compute the function $g(t)$. The first step is to obtain $\mathcal{E}_{(t+\epsilon,t)}$ via Eq. \eqref{intermediateE}. As suggested in the point (i) above, it is useful to employ the ``${\rm vec}$'' operation to obtain the inverse. We have
\begin{equation}
\mathbf{E}_{(t,0)}{\rm vec}(\rho)\equiv{\rm vec}[\mathcal{E}_{(t,0)}(\rho)]=\{[1-p(t)]\mathbb{I}_4+p(t)\sigma_z\otimes\sigma_z\}{\rm vec}(\rho),
\end{equation}
where we have used the property ${\rm vec}(ABC)=(C^{\rm t}\otimes A){\rm vec}(B)$ (cf. \cite{vec1,vec2}), and $\mathbb{I}_k$ stands for the $k\times k$ identity matrix. Therefore,
\begin{equation}
\mathbf{E}_{(t,0)}=\{[1-p(t)]\mathbb{I}_4+p(t)\sigma_z\otimes\sigma_z\}={\rm diag}[1,1-2p(t),1-2p(t),1],
\end{equation}
here ``${\rm diag}(a_1,a_2,\ldots,a_N)$'' denotes the diagonal matrix with entries $a_1,a_2,\ldots,a_N$. Hence,
\begin{equation}
\mathbf{E}_{(t+\epsilon,t)}=\mathbf{E}_{(t+\epsilon,0)}\mathbf{E}^{-1}_{(t,0)}={\rm diag}\left[1,\frac{1-2p(t+\epsilon)}{1-2p(t)},\frac{1-2p(t+\epsilon)}{1-2p(t)},1\right].
\end{equation}
Now, as commented in point (i) above, we have
\begin{equation}\label{EjemploRHP11}
{\rm vec}\left\{\left[\mathcal{E}_{(t+\epsilon,t)}\otimes\mathds{1}\right]\left(|\Phi\rangle\langle\Phi|\right)\right\}=U_{2\leftrightarrow 3}\left[\mathbf{E}_{(t+\epsilon,t)}\otimes\mathbb{I}_4\right]U_{2\leftrightarrow 3}{\rm vec}(|\Phi\rangle\langle\Phi|).
\end{equation}
In this case $U_{2\leftrightarrow 3}=\mathbb{I}_2\otimes \left(\begin{smallmatrix}
1  & 0 & 0 & 0 \\
0  & 0 & 1 & 0 \\
0  & 1 & 0 & 0 \\
0  & 0 & 0 & 1
\end{smallmatrix}\right) \otimes \mathbb{I}_2$ and $|\Phi\rangle=\tfrac{1}{\sqrt{2}}(1,0,0,1)^{\rm t}$, so that after some straightforward algebra, Eq. \eqref{EjemploRHP11} reads
\begin{equation}
{\rm vec}\left\{\left[\mathcal{E}_{(t+\epsilon,t)}\otimes\mathds{1}\right]\left(|\Phi\rangle\langle\Phi|\right)\right\}=\frac{1}{2}\left[1,0,0,\tfrac{1-2p(t+\epsilon)}{1-2p(t)},0,0,0,0,0,0,0,0,\tfrac{1-2p(t+\epsilon)}{1-2p(t)},0,0,1\right]^{\rm t}.
\end{equation}
By ``devectorizing'', i.e. writing this vector as the corresponding $4\times4$ matrix and computing the trace norm we immediately obtain
\begin{equation}
\left\|\left[\mathcal{E}_{(t,t_1)}\otimes \mathds{1}\right]\left(|\Phi\rangle\langle\Phi|\right)\right\|_1=\left|\tfrac{p(t)-p(\epsilon + t)}{1 - 2 p(t)}\right|+\left|\tfrac{1-p(t)-p(\epsilon + t)}{1 - 2 p(t)}\right|.
\end{equation}
Finally, by expanding at first order $p(t+\epsilon)\simeq p(t)+p'(t)\epsilon$, the limit in Eq. \eqref{gfunction} can be easily computed to arrive at
\begin{equation}
g(t)=\left|\tfrac{p'(t)}{1 - 2 p(t)}\right|-\tfrac{p'(t)}{1 - 2 p(t)}=\left\{\begin{array}{ll}
0 & \mbox{ if }  \left(\tfrac{p'(t)}{1 - 2 p(t)}\right)\geq0,\\
-\tfrac{2p'(t)}{1 - 2 p(t)} & \mbox{ if }   \left(\tfrac{p'(t)}{1 - 2 p(t)}\right)<0.\end{array}\right.
\end{equation}
Thus, given the function $p(t)$ and some interval $I$, with this result one immediately calculates $\mathcal{N}_{\rm RHP}^I$ or $\mathscr{D}_{\rm NM (RHP)}^I$.
\end{ex}

As aforementioned, the measure of non-Markovianity $\mathcal{N}_{\rm RHP}$, Eq. \eqref{CJI}, was first introduced in \cite{nuestro}. The more discriminative degree Eq. \eqref{CJdegree} is a variant based on the same ideas as \cite{normalizationChinos}, where the normalization problem was further analyzed.  Examples where this measure is studied can be found in \cite{Br-Q-Bos5,Br-Q-Bos7,Br-Q-Bos17,PinjaPreprint2,FrequencyXu,Anomalous-non-Markovian,Addis2} for qubits coupled to bosonic environments, in \cite{CJ1} for more general spin systems coupled to bosonic environments, in \cite{Br-Q-dlevel1,Br-Q-dlevel2} for qubits coupled to other $d-$level systems, in \cite{HRP,Mauro2ambientes,KnobMarkovianity} for qubits interacting with composite environments, in \cite{PRAPolonia,Vacch-NJP,LoFranco} for classical stochastic dynamics and in \cite{Vacch-NJP} for the so-called semi-Markov quantum processes. In addition, the application of the Choi-Jamio{\l}kowski criterion to study the complete positivity of intermediate dynamics for some specific examples is considered in \cite{CJ2} as well.

\subsection{Decay rates measures}

Since a Markovian dynamics is characterized by generators with the form of Eq. \eqref{diffMarkov}, in \cite{MichaelHall} Hall, Cresser, Li and Andersson proposed a measure of non-Markovianity focused on properties of the generator. Let us consider some dynamical evolution given by its generator,
\begin{eqnarray}\label{ACHgenerator}
\frac{d\rho(t)}{dt}=\mathcal{L}_t\left[\rho(t)\right]=&-&{\rm i}[H(t),\rho(t)]\\
&+&\sum_{k,\ell} c_{k\ell}(t)\left[V_k(t)\rho(t)W_\ell^\dagger(t)-\frac{1}{2}\{W_\ell^\dagger(t)V_k(t),\rho(t)\}\right]\nonumber.
\end{eqnarray}
In order to characterize its non-Markovianity, we may write $\mathcal{L}_t$ in an orthonormal basis $\{G_j\}_{j=0}^{d^2-1}$ with respect to the Hilbert-Schmidt product $\Tr(G_m^\dagger G_n)=\delta_{mn}$. More specifically, in \cite{MichaelHall} it is proposed to use a self-adjoint basis with $G_0=\mathbb{I}/\sqrt{d}$, so that $\{G_j\}_{j=1}^{d^2-1}$ can be taken to be the (normalized) generators of the $\mathfrak{su}(d)$ algebra. Thus, by expanding every operator of the dissipative part of the generator,
\begin{align}
V_k(t)&=\sum_{m}v_{km}(t)G_m,\quad v_{km}(t)=\Tr[G_mV_k(t)],\\
W_k(t)&=\sum_{n}w_{kn}(t)G_n,\quad w_{kn}(t)=\Tr[G_mW_k(t)].
\end{align}
Introducing this in Eq. \eqref{ACHgenerator}, one obtains
\begin{equation}
\mathcal{L}_t\left[\rho(t)\right]=-{\rm i}[H(t),\rho(t)]+\sum_{m,n} \tilde{c}_{mn}(t)\left[G_m\rho(t)G_n-\frac{1}{2}\{G_n G_m,\rho(t)\}\right],
\end{equation}
where $\tilde{c}_{mn}(t)=\sum_{k,\ell}w_{\ell n}^\ast(t) c_{k\ell}(t) v_{km}(t)$ forms a Hermitian matrix, $\tilde{c}_{mn}(t)=\tilde{c}_{nm}^\ast(t)$,  because $\mathcal{L}_t$ preserves the Hermiticity of $\rho$. Therefore, this matrix is diagonalized via some unitary operation, $\tilde{c}_{mn}(t)=\sum_j u_{mj}(t)\gamma_j(t) u_{nj}^\ast(t)$ and the generator can be rewritten in the form
\begin{equation}
\mathcal{L}_t\left[\rho(t)\right]=-{\rm i}[H(t),\rho(t)]+\sum_{j=1}^{d^2-1} \gamma_j(t)\left[L_j(t)\rho(t)L_j^\dagger(t)-\frac{1}{2}\{L_j^\dagger(t) L_j(t),\rho(t)\}\right]
\end{equation}
with $L_j(t)=\sum_m u_{mj}(t)G_m$, keeping orthonormality $\Tr[L_i^\dagger(t)L_j(t)]=\delta_{ij}$. Note that since the eigenvalues $\gamma_j(t)$ are independent of the basis, this form is unique (up to degeneracy). Now, Hall, Cresser, Li and Andersson define some functions of the eigenvalues (\emph{canonical decay rates}) $\gamma_j(t)$,
\begin{equation}\label{fj(t)}
f_j(t):=\max\{-\gamma_j(t),0\}.
\end{equation}
Because of Theorem \ref{KossLindTheo}, every $f_j(t)$ vanishes at any time if and only if the evolution is Markovian. Therefore the functions $f_j(t)$ can be used to construct a measure of non-Markovianity. For example, defining $f(t):=\sum_{j=1}^{d^2-1} f_j(t)$, for a (bounded) time interval $I$,
\begin{equation}
\mathcal{N}_{\gamma}^I:=\int_I f(t)dt,
\end{equation}
is a measure of non-Markovianity. Actually, it can be proven \cite{MichaelHall} that $f(t)=\frac{d}{2}g(t)$ [see Eq. \eqref{gfunction}], so this quantity is proportional to $\mathcal{N}_{\rm RHP}^I$, Eq. \eqref{CJI},
\begin{equation} \label{ACH-RHP}
\mathcal{N}_{\gamma}^I=\frac{d}{2}\mathcal{N}_{\rm RHP}^I.
\end{equation}

Interestingly, this approach also suggest a discrete measure, by computing $F_j^I=\int_I f_j(t)dt$, a \emph{non-Markovianity index} can be defined by the rule
\begin{equation}
\mathcal{N}_{\rm index}^I:=\sum_{j=1}^{d^2-1}\chi(F_j^I),
\end{equation}
i.e. the number of non-zero $F_j^I$'s in the interval $I$.

\begin{ex} Consider the evolution of a qubit given by the following master equation
\begin{align}\label{exACH}
\frac{d \rho(t)}{dt}=\mathcal{L}_t[\rho(t)]=-{\rm i}\omega[\sigma_z,\rho(t)]&+\gamma_{-}(t)\left[\sigma_-\rho(t)\sigma_+-\tfrac{1}{2}\{\sigma_+\sigma_-,\rho(t)\}\right]\nonumber\\
&+\gamma_{z}(t)\left[\sigma_z\rho(t)\sigma_z-\rho(t)\right],
\end{align}
subject to the conditions $\int_{t_0}^t\gamma_{-}(s)ds\geq0$ and $\int_{t_0}^t\gamma_{z}(s)ds\geq0$ to ensure the complete positivity of the dynamical map $\mathcal{E}_{(t,t_0)}$. Let us compute the functions $g(t)$ and $f(t)$. For the first one, we use the formula in terms of the generator $\mathcal{L}_t$, Eq. \eqref{gfunctionL}. By computing the eigenvalues of $\left[\mathds{1}+\epsilon(\mathcal{L}_t\otimes\mathds{1})\right](|\Phi\rangle\langle\Phi|)$ and expanding each of them to the first order in $\epsilon$ we obtain
\begin{equation}
\left\Vert\left[\mathds{1}+\epsilon(\mathcal{L}_t\otimes\mathds{1})\right](|\Phi\rangle\langle\Phi|)\right\Vert_1=\tfrac{1}{2}|\gamma_{-}(t)\epsilon|+|\gamma_{z}(t)\epsilon+\mathcal{O}(\epsilon^2)|+\left|1-[\tfrac{1}{2}\gamma_{-}(t) +\gamma_{z}(t)]\epsilon+\mathcal{O}(\epsilon^2)\right|.
\end{equation}
Thus, the limit of Eq. \eqref{gfunctionL} is readily computed,
\begin{equation}
g(t)=\tfrac{1}{2}[|\gamma_{-}(t)|-\gamma_{-}(t)]+|\gamma_{z}(t)|-\gamma_{z}(t).
\end{equation}
Now, in order to find the functions $f_j(t)$, Eq. \eqref{fj(t)}, we have to write Eq. \eqref{exACH} in a orthonormal basis with respect to the Hilbert-Schmidt product. However, since $\sigma_\pm=\tfrac{1}{2}(\sigma_x\pm{\rm i}\sigma_y)$ and because of the orthogonality of the Pauli matrices, $\big\{\tfrac{1}{\sqrt{2}}\mathbb{I}_2,\sigma_-,\sigma_+,\tfrac{1}{\sqrt{2}}\sigma_z\big\}$ forms an orthonormal basis. Thus, the canonical decay rates are $\gamma_-(t)$ and $2\gamma_z(t)$. By noting that $\max\{-\gamma_j(t),0\}=\tfrac{1}{2}[|\gamma_j(t)|-\gamma_j(t)]$ we obtain
\begin{equation}
f(t)=\sum_{j} f_j(t)=\tfrac{1}{2}[|\gamma_-(t)|-\gamma_-(t)]+|\gamma_z(t)|-\gamma_z(t).
\end{equation}
Therefore, $g(t)=f(t)$ as expected in this case since $d=2$. For Markovian evolution $\gamma_-(t)\geq0$, $\gamma_z(t)\geq0$ for all $t$ and $g(t)=f(t)=0$.
\end{ex}

Other examples where this measure is applied can be found in \cite{LoFranco,MichaelHall}. See also \cite{SunExp} for an experimental proposal to probe non-Markovianity by negative decay rates.

\subsection{Hierarchical $k$-divisibility degrees}\label{sec:DariuszSabrina}
Recently, Chru\'{s}ci\'{n}ski and Maniscalco have proposed a hierarchical way to assess non-Markovianity \cite{hierarchical}. Their approach, based on the concept of $k$-divisibility, is interesting as it provides a way to define some kind of maximally non-Markovian dynamics. Basically, a family of dynamical maps, $\{\mathcal{E}_{(t_2,t_1)},t_2\geq t_1\geq t_0\}$, is $k$-divisible, if $\mathcal{E}_{(t_2,t_1)}\otimes\mathds{1}_k$ is a positive map for all $t_2\geq t_1\geq t_0$ (here $\mathds{1}_k$ denotes the identity map acting on the space of $k\times k$ matrices). Therefore, if the dimension of the quantum system is $d$, a $k$-divisible process with $k\geq d$, is what in this work has been called divisible or Markovian process (see Definition \ref{def-Marko}). The 1-divisible processes are the P-divisible processes as introduced in Definition \ref{defpre-Marko2}, and the 0-divisible processes are processes where $\mathcal{E}_{(t_2,t_1)}$ is not a positive operator for some $t_1$ and $t_2\geq t_1$.

Moreover, analogously to Theorems \ref{theoCPTcontractions} and \ref{contractionMarkov}, we have that a process is $k$-divisible if and only if $\tilde{\sigma}_k(\tilde{\Delta},t):=\frac{d}{dt}\big\|\left[\mathcal{E}_{(t,t_0)}\otimes\mathds{1}_k\right]\tilde{\Delta}\big\|_1\leq0$ for every Helstrom matrix $\tilde{\Delta}=q\rho_{1A}-(1-q)\rho_{2A}$ with an ancillary space of dimension $k$. In similar fashion to Eq. \eqref{NHelstrom}, Chru\'{s}ci\'{n}ski and Maniscalco define a set of degrees to quantify departure from $k$-divisibility for $t\in I$,
\begin{equation}\label{Sabrina}
\mathscr{D}_{{\rm ND} (k)}^I:=\sup_{\tilde{\Delta}}\frac{N_{k}^{+}(\tilde{\Delta},I)}{|N_{k}^{-}(\tilde{\Delta},I)|},
\end{equation}
where $N_\pm^I(\tilde{\Delta},t):=\int_{t\in I,\tilde{\sigma}\gtrless 0}dt\tilde{\sigma}_k(\Delta,t)$, and the subindex ``ND'' stands for non-divisibility.

Since $\mathcal{E}_{(t,t_0)}$ is completely positive for any final time $t$, it is easy to prove that $N_{k}^{+}(\tilde{\Delta},I)\leq |N_{k}^{-}(\tilde{\Delta},I)|$ \cite{hierarchical}, therefore $\mathscr{D}_{{\rm ND} (k)}^I\in [0,1]$ for all $k$. Moreover, as $k$ increases, so does the dimension of the space in the optimization problem Eq. \eqref{Sabrina}, and hence it is clear that
\begin{equation}
0\leq \mathscr{D}_{\rm ND (1)}^I\leq \ldots \leq \mathscr{D}_{{\rm ND} (d)}^I \leq 1.
\end{equation}
In this equation, only $\mathscr{D}_{{\rm ND} (d)}^I$ is a degree of non-Markovianity as defined in this work. The other quantities are zero for non-Markovian but $k$-divisible $(k<d)$ dynamics. This hierarchy of degrees of non-divisibility is particularly useful to try a definition of maximally non-Markovian dynamics. Indeed, Chru\'{s}ci\'{n}ski and Maniscalco propose to call ``maximally non-Markovian dynamics'' to those that  $\mathscr{D}_{{\rm ND} (1)}^I=1$, and consequently $\mathscr{D}_{{\rm ND} (2)}^I=\mathscr{D}_{{\rm ND} (3)}^I=\ldots=\mathscr{D}_{{\rm ND} (d)}^I=1$. A particular example of this kind dynamics for a qubit is the one generated by the master equation $\frac{d}{dt}\rho=\gamma(t)(\sigma_z\rho\sigma_z-\rho)$, for an interval $I$ such that $\gamma(t)$ is periodic in $I$. For instance, $\gamma(t)=\sin(t)$ or $\gamma(t)=\tan(t)$ in $t\in [0,2\pi]$. Interestingly, for these two examples the Choi-Jamiolkowski measure Eq. \eqref{CJdegree} provides different values. We obtain $\mathscr{D}_{\rm NM (RHP)}^I=0.758$ and $\mathscr{D}_{\rm NM (RHP)}^I=0.803$, for $\gamma(t)=\sin(t)$ and $\gamma(t)=\tan(t)$ respectively.

\section{Witnesses of Quantum non-Markovianity}\label{section:witness}

In this section we revise the different ways to detect non-Markovianity via witnesses. A \emph{witness of non-Markovianity} is a quantity that vanishes for all Markovian dynamics (see also \cite{TrazaNoHermitica,DarekKoss1,DarekKoss2}), but it may also vanish for some non-Markovian dynamics. Thus, when a witness of non-Markovianity gives a non-zero value, we are sure that the dynamics is non-Markovian.

In general, we can classify the witnesses of non-Markovianity that have been presented in the literature according to two guiding principles. There are witnesses based on monotonic quantities under completely positive maps, and based on monotonic quantities under local completely positive maps. In the following we review several proposals in these two classes and illustrate their use with a simple example.

\subsection{Witnesses based on monotonicity under completely positive maps}

\subsubsection{Trace distance and the BLP quantifier} \label{secTracedistance}

If we consider the unbiased situation in the two-state discrimination problem, $q=1/2$ and the Helstrom matrix reads as $\Delta=(\rho_1-\rho_2)/2$, where we have neglected the possible presence of ancillary systems. Thus, the trace norm of $\Delta$ becomes the \emph{trace distance} between states $\rho_1$ and $\rho_2$,
\begin{equation}
D_{1}(\rho_1,\rho_2):=\frac{1}{2}\|\rho_1-\rho_2\|_1.
\end{equation}
Analogously to Eqs. \eqref{tildesigma} and \eqref{tildesigmaInt}, we write,
\begin{equation}\label{sigmaDist}
\sigma(\rho_1,\rho_2,t):=\frac{dD_1[\rho_1(t),\rho_2(t)]}{dt}:=\lim_{\epsilon\rightarrow0^+}\frac{D_1[\rho_1(t+\epsilon),\rho_2(t+\epsilon)]-D_1[\rho_1(t),\rho_2(t)]}{\epsilon},
\end{equation}
and for some interval $I$,
\begin{equation}\label{Intdist}
\int_{t\in(t_1,t_2),\sigma>0}dt\sigma(\rho_1,\rho_2,t).
\end{equation}
If this quantity is not zero for some pair of states $\rho_1$ and $\rho_2$, we are sure the dynamics is non-Markovian in $I$, as it is a lower bound to the non-Markovianity measure $\mathcal{N}_{\rm H}^I$, Eq. \eqref{NHelstrom}. Particularly, we may be interested in finding the largest value of Eq. \eqref{Intdist} in the time interval $(0,\infty)$. To this end, Breuer, Laine and Piilo \cite{BrLaPi1,BrLaPi2} define the quantifier
\begin{equation}
\mathcal{N}_{\rm BLP}:=\max_{\rho_1,\rho_2}\int_{\sigma>0}dt\sigma(\rho_1,\rho_2,t).
\end{equation}

\begin{ex} Consider the following master equation of a qubit system
\begin{equation}\label{PureDephMaster}
\frac{d\rho(t)}{dt}=\gamma(t)[\sigma_z\rho(t)\sigma_z-\rho(t)],
\end{equation}
with $\int_{t_0}^t\gamma(s)ds\geq0$ for completely positive dynamics. This equation can be easily integrated, by writing $\rho(0)=\big(\begin{smallmatrix}
\rho_{11} & \rho_{12}\\
\rho_{21}& \rho_{22}
\end{smallmatrix}\big)$, we obtain ($t_0=0$  without loss of generality)
\begin{equation}\label{PureDephIntegrated}
\rho(t)=\mathcal{E}_{(t,0)}[\rho(0)]=\begin{bmatrix}
\rho_{11} & \rho_{12}R(t)\\
\rho_{21}R(t) & \rho_{22}
\end{bmatrix}, \quad \text{with } R(t)={\rm e}^{-2\int_{0}^t\gamma(s)ds},
\end{equation}
Note that $0\leq R(t)\leq 1$. Let us compute the trace distance between two different initial states, for example $\rho_1=\tfrac{1}{2}\big(\begin{smallmatrix}
1 & 1\\
1 & 1
\end{smallmatrix}\big)$ and $\rho_2=\tfrac{1}{2}\big(\begin{smallmatrix}
1 & -1\\
-1 & 1
\end{smallmatrix}\big)$. Because of Eq. \eqref{PureDephIntegrated} we immediately obtain
\begin{equation}
D_1[\rho_1(t),\rho_2(t)]=\frac{1}{2}\|\rho_1(t)-\rho_2(t)\|_1=\left\|\Big[\begin{smallmatrix}
0 & R(t)\\
R(t) & 0
\end{smallmatrix}\Big]\right\|_1=|R(t)|=R(t),
\end{equation}
so that,
\begin{equation}
\sigma(\rho_1,\rho_2,t)=\frac{dD_1[\rho_1(t),\rho_2(t)]}{dt}=-2\gamma(t)R(t).
\end{equation}
Therefore, if $\sigma(\rho_1,\rho_2,t)>0$ for some $t$, then $\gamma(t)<0$, and the dynamics is non-Markovian.
\end{ex}

The use of the trace distance to witness non-Markovianity was originally proposed in \cite{BrLaPi1} and further analyzed in \cite{BrLaPi2}. Due to its simplicity and intuitive physical interpretation, it has been applied to detect non-Markovian features in dynamics of qubit \cite{Addis2,Br-Q-Bos1,Br-Q-Bos2,Br-Q-Bos3,Br-Q-Bos4,Br-Q-Bos5,Br-Q-Bos6,Br-Q-Bos7,Br-Q-Bos8,Br-Q-Bos9,Br-Q-Bos10,Br-Q-Bos11,Br-Q-Bos12,Br-Q-Bos13,Br-Q-Bos14,Br-Q-Bos15,Br-Q-Bos16,Br-Q-Bos17,Br-Q-Bos18,JQIS,Br-Q-Bos19,PinjaPreprint1,NoMarkovZeno,PinjaPreprint2,FrequencyXu,Anomalous-non-Markovian,Steffen,XuEnesimo} and qutrit systems \cite{Br-Qtr-Bos} coupled to bosonic environments, qubits coupled to other finite dimensional systems \cite{Br-Q-dlevel1,Br-Q-dlevel2,Br-Q-dlevel3}, and to composite \cite{Mauro2ambientes,KnobMarkovianity,Br-CompoEnv1,Br-CompoEnv2,Br-CompoEnv3} and chaotic \cite{Br-Chaotic,Br-Chaotic1,Br-Chaotic2} environments. It has also been employed to analyze memory-kernel master equations \cite{PhenoMasterEq,Vacch-Piece}, quantum semi-Markov process \cite{Vacch-NJP}, classical noise \cite{Vacch-NJP,PRAPolonia,ShibataPRA,LoFranco,ManiscalcoParisClassical,ComparisonBLPMutual}, in fermionic systems \cite{Chancellor} and collisional models \cite{TraceCollisional}, and to study exciton-phonon dynamics in energy transfer of photosynthetic complexes \cite{PatrickAlan,KoNaOlaya}. Moreover this witness has been implemented experimentally within a linear optics set up \cite{ExperimentoPiilo1,ExperimentoPiilo2,ExperimentoPiilo3,RLoFranco,BreuerMax2yExperimento}. On the other hand, some connections have been found between the non-monotonic behavior of the trace distance and geometric phases \cite{Geometricphases}, Loschmidt echo \cite{Loschmidt,Loschmidt2}, dynamical recovering of the quantum coherence by applying local operations \cite{DinamicalRecovering}, and the appearance of system-environment correlations \cite{LauraMauroBassano}. In this regard, this witness has also been proposed as a tool to detect the presence of initial system-environment correlations \cite{InCorr1,InCorr2,InCorr3,InCorr4,CesarKavanLaura,Br-Q-Bos11,Br-Q-Bos14,Br-Q-Bos16,Dajka,Wissmann}.

While very efficient under certain conditions, there are some non-Markovian processes which cannot be witnessed by the trace distance, for example those where the non-Markovianity is encoded just in the ``non-unital part'' of the dynamics \cite{NoUnital}. This part corresponds to the affine vector $\bm{c}_{(t,t_0)}$ when the dynamics is visualized in the Bloch space, see Eq. \eqref{affineBloch}. Necessary and sufficient conditions for trace distance to witness non-Markovianity can be found in \cite{MichaelHall}.

As a matter of curiosity, the trace distance has been also adapted to measure the degree of non-Markovianity of musical compositions \cite{Musical}.

\subsubsection{Fidelity}
The fidelity $F(\rho_1,\rho_2)$ between two quantum states $\rho_1$ and $\rho_2$ is a generalization of the transition probability $|\langle\psi_1|\psi_2\rangle|^2$ between two pure states $\ket{\psi_1}$ and $\ket{\psi_2}$, to density matrices. Specifically, the fidelity is defined \cite{Uhlmann,Jozsa} by the equation
\begin{equation}\label{fidelityMax}
F(\rho_1,\rho_2):=\max_{|\Psi_1\rangle,|\Psi_2\rangle}|\langle\Psi_1|\Psi_2\rangle|^2.
\end{equation}
Here, $\ket{\Psi_1}$ and $\ket{\Psi_2}$ are two purifications of $\rho_1=\Tr_{\rm A}(|\Psi_1\rangle\langle\Psi_1|)$ and $\rho_2=\Tr_{\rm A}(|\Psi_2\rangle\langle\Psi_2|)$, where $\Tr_{\rm A}$ denotes the partial trace on some ancillary system A, and the maximum is taken over the all possible purifications (see \cite{NC00} for a pedagogical introduction \cite{footnote2}). Uhlmann \cite{Uhlmann} solved the optimization problem obtaining
\begin{equation}
F(\rho_1,\rho_2)=\|\sqrt{\rho_1}\sqrt{\rho_2}\|_1^2=\left[\Tr\sqrt{\sqrt{\rho_1}\rho_2\sqrt{\rho_1}}\right]^2.
\end{equation}

Among several properties, the fidelity is monotonic under complete positive maps $\mathcal{E}$,
\begin{equation}\label{monoFidelity}
F[\mathcal{E}(\rho_1),\mathcal{E}(\rho_2)]\geq F(\rho_1,\rho_2),
\end{equation}
reaching the equality if and only if the completely positive map is unitary $\mathcal{E}(\cdot)=U(\cdot)U^\dagger$ \cite{Molnar}. Thus, the fidelity is monotonically increasing for Markovian evolutions, and therefore it may be used to witness non-Markovian behavior.

\begin{ex} Consider again the simple model of pure dephasing, Eqs. \eqref{PureDephMaster} and \eqref{PureDephIntegrated}. Again, for the initial states $\rho_1=\tfrac{1}{2}\big(\begin{smallmatrix}
1 & 1\\
1 & 1
\end{smallmatrix}\big)$ and $\rho_2=\tfrac{1}{2}\big(\begin{smallmatrix}
1 & -1\\
-1 & 1
\end{smallmatrix}\big)$, a straightforward computation of the fidelity gives
\begin{equation}
F[\rho_1(t),\rho_2(t)]=\left[\Tr\sqrt{\sqrt{\rho_1(t)}\rho_2(t)\sqrt{\rho_1(t)}}\right]^2=1-R^2(t),
\end{equation}
so that
\begin{equation}
\frac{dF[\rho_1(t),\rho_2(t)]}{dt}=-2R'(t)R(t)=4\gamma(t)R^2(t).
\end{equation}
Hence if the fidelity decreases at some time $t$, then $\gamma(t)<0$ and the dynamics is non-Markovian.
\end{ex}

In Ref. \cite{FidelityGauss} the approach of \cite{BrLaPi1,BrLaPi2}, originally proposed for the trace distance (see previous section), is reconsidered with the so-called \emph{Bures distance} \cite{Hubner}:
\begin{equation}\label{bures}
D_B(\rho_1,\rho_2):=\sqrt{2\left[1-\sqrt{F(\rho_1,\rho_2)}\right]}.
\end{equation}
Since the authors of \cite{FidelityGauss} aim at quantifying non-Markovianity in Gaussian states of harmonic oscillators, the use of the fidelity instead the trace distance is more convenient because a closed formula for the latter for Gaussian states is still lacking. Other examples can be found in \cite{Doug,GalveZambrini,ChinosFisher1}. Regarding Gaussian states, an alternative approach to witness non-Markovianity is suggested in \cite{GaussianCovMat}.

A different witness in terms of fidelity was previously proposed in \cite{Usha-Devi1}, however that was only able to detect deviations from time-homogeneous Markov processes, i.e. quantum dynamical semigroups where $\mathcal{E}_{(t_2,t_1)}=\mathcal{E}_{(t_2-t_1)}$ for every $t_1$ and $t_2$. See \cite{DavidWolf} for another work exclusively focused on time-homogeneous dynamics.

\subsubsection{Quantum relative entropies}

Another similar witness is constructed with the (von Neumann) \emph{relative entropy} between two quantum states $\rho_1$ and $\rho_2$,
\begin{equation}
S(\rho_1 \| \rho_2):=\Tr(\rho_1\log\rho_1)-\Tr(\rho_1\log\rho_2).
\end{equation}
Despite the relative entropy neither being symmetric $S(\rho_1 \| \rho_2)\neq S(\rho_2 \| \rho_1)$, nor satisfying the triangle inequality, it is often intuited as a distance measure because $S(\rho_1 \| \rho_2)\geq0$, vanishing if and only if $\rho_1=\rho_2$ (Klein's inequality \cite{NC00}). Moreover, if the intersection of the kernel of $\rho_2$ with the support of $\rho_1$ is non-trivial, then $S(\rho_1 \| \rho_2)$ becomes infinity.

Analogously to the Bures and the trace distance, the quantum relative entropy is monotonic under completely positive and trace preserving maps $\mathcal{E}$,
\begin{equation} \label{QREmono}
S[\mathcal{E}(\rho_1) \| \mathcal{E}(\rho_2)]\leq S(\rho_1 \| \rho_2).
\end{equation}
The proof of this result was fist given by Lindblad \cite{LindbladEntropy} for finite dimensional systems, and Uhlmann \cite{UhlmannEntropy} extend it to the general case (see also \cite{Hayashi} and \cite{RuskaiEntropy}).

Therefore the quantum relative entropy between any two states is monotonically decreasing with time in a Markovian process, and any increment of it at some time instant reveals the non-Markovian character of the dynamics.

\begin{ex} For the model of pure dephasing, Eqs. \eqref{PureDephMaster} and \eqref{PureDephIntegrated}, and initial states $\rho_1=\tfrac{1}{2}\big(\begin{smallmatrix}
1 & 1\\
1 & 1
\end{smallmatrix}\big)$ and $\rho_2=\tfrac{1}{2}\big(\begin{smallmatrix}
1 & -1\\
-1 & 1
\end{smallmatrix}\big)$, the quantum relative entropy becomes
\begin{equation}
S[\rho_1(t)\|\rho_2(t)]=\Tr[\rho_1(t)\log\rho_1(t)]-\Tr[\rho_1(t)\log\rho_2(t)]=R(t)\log\left[\frac{1+R(t)}{1-R(t)}\right],
\end{equation}
so that its derivative is
\begin{equation}
\frac{dS[\rho_1(t)\|\rho_2(t)]}{dt}=-2\gamma(t)R(t)\left\{\frac{2R(t)}{1-R^2(t)}+\log\left[\frac{1+R(t)}{1-R(t)}\right]\right\}.
\end{equation}
Since $0\leq R(t)\leq1$, everything multiplying $\gamma(t)$ in the above equation is negative. Hence, an increment in the quantum relative entropy at some $t$ implies $\gamma(t)<0$ and non-Markovianity.
\end{ex}

The use of the quantum relative entropy to witness non-Markovianity was originally proposed in \cite{BrLaPi2}. In \cite{TrazaNoHermitica} it is suggested to use more general relative entropies due to Renyi \cite{RenyiMonotonicity} and Tsallis \cite{TsallisMonotonicity} for the same task. In this regard, \cite{Kastoryano} enumerates several distances fulfilling the monotonicity condition. Additionally \cite{Usha-Devi2} proposed to use the monotonicity of the relative entropy to detect the presence of initial system-environment correlations.

\subsubsection{Quantum Fisher information}

Following \cite{Holevo,BCaves} (see also \cite{WisemanMilburn}) the \emph{quantum Fisher information} can be defined as the infinitesimal Bures distance \eqref{bures} between two quantum states. For simplicity, assume some one-parametric family of quantum states $\rho_\theta$, then
\begin{equation}\label{FisherBures}
D_B^2(\rho_\theta,\rho_{\theta+\delta\theta})=\frac{1}{4}\mathcal{F}(\rho_\theta)(\delta\theta)^2+\mathcal{O}[(\delta\theta)^3],
\end{equation}
where $\mathcal{F}(\rho_\theta)$ is the so-called quantum Fisher information of the family $\rho_\theta$. Equivalently, we write
\begin{equation}\label{Fisher}
\mathcal{F}(\rho_\theta):=4\lim_{\delta\theta\rightarrow0}\left[\frac{D_B(\rho_\theta,\rho_{\theta+\delta\theta})}{\delta\theta}\right]^2.
\end{equation}
Thus, the quantum Fisher information of $\rho_\theta$ measures the sensitivity of the Bures distance when $\theta$ is varied. In turn, this can be interpreted as the information about $\theta$ which is contained in the family $\rho_\theta$, in such a way that if $\rho_\theta$ does not depend on $\theta$, $\mathcal{F}(\rho_\theta)=0$. We will come back to this point later.

Additionally, the quantum Fisher information admits other different but equivalent definitions \cite{Holevo,BCaves,WisemanMilburn,Helstrom}. For example \cite{BCaves}, it can be defined as the maximum Fisher information of classical probabilities $p(x|\theta)=\Tr(\Pi_x\rho_\theta)$, where the optimization is made over all possible POVMs $\{\Pi_x\}$,
\begin{equation}
\mathcal{F}(\rho_\theta):=\max_{\Pi_x} \mathcal{F}[p(x|\theta)].
\end{equation}
Recall that the Fisher information of a probability distribution $p(x|\theta)$ is defined as
\begin{equation}
\mathcal{F}[p(x|\theta)]:=\int \frac{1}{p(x|\theta)}\left[\frac{\partial p(x|\theta)}{\partial\theta}\right]^2dx.
\end{equation}
Another definition is given in terms of the so-called \emph{symmetric logarithmic derivative} operator $L$, which is defined via the implicit equation
\begin{equation}
\frac{d\rho_\theta}{d\theta}:=\frac{1}{2}\left(L\rho_\theta+\rho_\theta L\right),
\end{equation}
and depends on the particular form of $\rho_\theta$, $L=L(\rho_\theta)$. The quantum Fisher information is given by the variance of this operator in the family $\rho_\theta$ \cite{Helstrom},
\begin{equation}\label{FisherSLD}
\mathcal{F}(\rho_\theta):=\Tr\left[L^2(\rho_\theta)\rho_\theta\right].
\end{equation}
The equivalence between Eqs. \eqref{Fisher} and \eqref{FisherSLD} can be found explicitly proven in \cite{Hubner}.

Going back to the problem of witnessing non-Markovianity, the quantum Fisher information is also monotonically decreasing under Markovian dynamics, as it cannot increase under completely positive maps. This can be showed directly from Eq. \eqref{Fisher}. Because of \eqref{monoFidelity}, the Bures distance (and its square) is monotonically decreasing under a completely positive $\mathcal{E}$, so that
\begin{equation}
\mathcal{F}[\mathcal{E}(\rho_\theta)]=4\lim_{\delta\theta\rightarrow0}\frac{D_B^2[\mathcal{E}(\rho_\theta),\mathcal{E}(\rho_{\theta+\delta\theta})]}{(\delta\theta)^2}\leq4\lim_{\delta\theta\rightarrow0}\frac{D_B^2(\rho_\theta,\rho_{\theta+\delta\theta})}{(\delta\theta)^2}=\mathcal{F}(\rho_\theta).
\end{equation}
The use of the quantum Fisher information to witness non-Markovianity is originally due to Lu, Wang and Sun in \cite{Sun}. These authors provided a proof of the monotonicity of the Fisher information by using the definition Eq. \eqref{FisherSLD}, and introduced a flow of quantum Fisher information by
\begin{equation}\label{FisherFlow}
\mathcal{I}_{\rm QFI}(t):=\frac{\partial \mathcal{F}[\rho_\theta(t)]}{\partial t}.
\end{equation}
Thus if $\mathcal{I}_{\rm QFI}(t)>0$ for some $t$, the time evolution is non-Markovian. Moreover, if the evolution is given by some master equation,
\begin{equation}
\frac{d\rho(t)}{dt}=\mathcal{L}_t\rho(t)=-{\rm i}[H(t),\rho(t)]+\sum_k\gamma_k(t)\left[V_k(t)\rho(t)V_k^\dagger(t)-\frac{1}{2}\{V_k^\dagger(t)V_k(t),\rho(t)\}\right],
\end{equation}
the quantum Fisher information flow can be written as
\begin{eqnarray}
\mathcal{I}_{\rm QFI}(t)&=&\sum_k \gamma_k(t)\mathcal{J}_k(t), \\ \mathcal{J}_k(t)&:=&-\Tr\left\{\rho_\theta(t)[L(\rho_\theta,t),V_k(t)]^\dagger[L(\rho_\theta,t),V_k(t)]\right\}\leq0.
\end{eqnarray}
Therefore, $\mathcal{I}_{\rm QFI}$ is negative if all $\gamma_k(t)\geq0$ in accordance with Theorem \ref{KossLindTheo}.

\begin{ex} Consider the family of states $\rho_\theta=\tfrac{1}{2}\big(\begin{smallmatrix}
1 & {\rm e}^{-{\rm i}\theta} \\
{\rm e}^{{\rm i}\theta} & 1
\end{smallmatrix}\big)$ which is typically generated by applying the phase shift $\theta$ to the state $\ket{\psi}=\tfrac{1}{\sqrt{2}}(1,1)^{\rm t}$. If $\rho_\theta$ is subject to the pure dephasing, Eqs. \eqref{PureDephMaster} and \eqref{PureDephIntegrated}, we can compute the Fisher information directly by expanding the squared Bures distance between $\rho_\theta(t)$ and $\rho_{\theta+\delta\theta}(t)$ up to the second order in $\delta\theta$ [Eq. \eqref{FisherBures}] . After some algebra we find
\begin{equation}\label{FisherEx}
D_B^2(\rho_\theta,\rho_{\theta+\delta\theta})=\frac{1}{4}R^2(t)(\delta\theta)^2+\mathcal{O}[(\delta\theta)^3]\Rightarrow \mathcal{F}(\rho_\theta)=R^2(t).
\end{equation}
Thus, the quantum Fisher information flow, Eq. \eqref{FisherFlow} is
\begin{equation}
\mathcal{I}_{\rm QFI}(t)=\gamma(t)\mathcal{J}(t), \quad \text{with } \mathcal{J}(t)=-4R^2(t),
\end{equation}
and $\mathcal{I}_{\rm QFI}(t)>0$ for some $t$ denotes $\gamma(t)<0$ and non-Markovianity.
\end{ex}

Other examples where the quantum Fisher information flow is computed can be found in the original reference \cite{Sun} and in \cite{Loschmidt,ChinosFisher1,ChinosFisher2}, where its possible relation with the Loschmidt echo was explored.

Notably, this witness of non-Markovianity may be relevant in the context of quantum parameter estimation. Specifically, the error (variance) of any (unbiased) estimation of the parameter $\theta$ is related to the quantum Fisher information through the quantum Cramer-Rao bound \cite{Holevo,BCaves,WisemanMilburn,Helstrom}:
\begin{equation}
(\Delta\theta)^2\geq\frac{1}{\mathcal{F}(\rho_\theta)}.
\end{equation}
Thus, an increment in $\mathcal{F}(\rho_\theta)$ could be linked with a increment of information about the parameter $\theta$. Nevertheless note that the quantum Fisher information provide just a lower bound to the error on $\theta$, and in fact there are cases where this bound is not achievable.

\subsubsection{Capacity measures}

In \cite{DarekSabrina} Bylicka, Chru\'{s}ci\'{n}ski and Maniscalco have suggested to use capacity measures to detect non-Markovianity. Specifically, given a complete positive and trace-preserving map $\mathcal{E}$ and some quantum state $\rho$, we introduce the mutual information between $\rho$ and $\mathcal{E}(\rho)$ via
\begin{equation}
I(\rho,\mathcal{E}):=S(\rho)+I_c(\rho,\mathcal{E}).
\end{equation}
Here $S(\rho)=-\Tr(\rho\log\rho)$ is the von Neumann entropy and $I_c(\rho,\mathcal{E})$ is the so-called quantum coherent information, defined as \cite{NC00},
\begin{equation}
I_c(\rho,\mathcal{E}):=S[\mathcal{E}(\rho)]-S\{[\mathcal{E}\otimes\mathds{1}](|\Psi_\rho\rangle\langle\Psi_\rho|)\},
\end{equation}
where $|\Psi_\rho\rangle\in\mathcal{H}\otimes\mathcal{H}_{\rm A}$ is a purification of $\rho=\Tr_{\rm A}(|\Psi_\rho\rangle\langle\Psi_\rho|)$. Remarkably, the quantity $S\{[\mathcal{E}\otimes\mathds{1}](|\Psi_\rho\rangle\langle\Psi_\rho|)\}$ does not depend on the particular choice of purification.
The quantum coherent information is monotonic under completely positive maps \cite{footnoteDataProcesing}
\begin{equation}
I_c(\rho,\mathcal{E}_2\mathcal{E}_1)\leq I_c(\rho,\mathcal{E}_1),
\end{equation}
and the same equation is satisfied for $I(\rho,\mathcal{E})$. Thus, in Ref. \cite{DarekSabrina} the following two capacity measures are proposed to witness non-Markovianity,
\begin{align}
C_{ea}\left[\mathcal{E}_{(t,t_0)}\right]&:=\sup_\rho I\left[\rho,\mathcal{E}_{(t,t_0)}\right],\label{CeaDefi}\\
Q\left[\mathcal{E}_{(t,t_0)}\right]&:=\sup_\rho I_c\left[\rho,\mathcal{E}_{(t,t_0)}\right],\label{QDefi}
\end{align}
The entanglement assisted capacity $C_{ea}$ sets a bound on the amount of classical information which can be transmitted along the dynamical process described by $\mathcal{E}_{(t,t_0)}$ when sender at $t_0$ and receiver at $t$ are allowed to share an unlimited amount of entanglement. Similarly, the capacity $Q$ provides the limit to the rate at which quantum information can be reliably sent by the quantum process (for a singe use).

\begin{ex} \label{ExCapacity} Let us calculate the capacity measures for the pure dephasing model, Eqs. \eqref{PureDephMaster} and \eqref{PureDephIntegrated}. It is immediate to check that the dynamical map in this case is given by
\begin{equation}\label{PureDephE}
\mathcal{E}_{(t,0)}[\rho]=\left[\tfrac{1+R(t)}{2}\right]\rho+\left[\tfrac{1-R(t)}{2}\right]\sigma_z\rho\sigma_z.
\end{equation}
It can be shown \cite{Addis2,MarkWilde} that the maximum for both measures $C_{ea}$ and $Q$ is reached for a maximally mixed state, $\rho=\mathbb{I}_2/2$, and then the required purification has to be a maximally entangled state, e.g. $|\Psi_\rho\rangle=|\Phi\rangle=\tfrac{1}{\sqrt{2}}(1,0,0,1)^{\rm t}$,
\begin{equation}\label{Ex1PHI}
[\mathcal{E}_{(t,0)}\otimes\mathds{1}](|\Psi_\rho\rangle\langle\Psi_\rho|)=\frac{1}{2}\left[\begin{smallmatrix}
1 & 0 & 0 & R(t) \\
0 & 0 & 0 & 0  \\
0 & 0 & 0 & 0  \\
R(t) & 0 & 0 & 1 \\
\end{smallmatrix}\right].
\end{equation}
Thus a direct computation leads to
\begin{align}
C_{ea}\left[\mathcal{E}_{(t,0)}\right]&=2\log 2+\left[\tfrac{1+R(t)}{2}\right]\log\left[\tfrac{1+R(t)}{2}\right]+\left[\tfrac{1-R(t)}{2}\right]\log\left[\tfrac{1-R(t)}{2}\right],\label{CeaEq}\\
Q\left[\mathcal{E}_{(t,0)}\right]&=\log2+\left[\tfrac{1+R(t)}{2}\right]\log\left[\tfrac{1+R(t)}{2}\right]+\left[\tfrac{1-R(t)}{2}\right]\log\left[\tfrac{1-R(t)}{2}\right]\label{QEq}.
\end{align}
Therefore both quantities have the same derivative
\begin{equation}
\frac{dC_{ea}\left[\mathcal{E}_{(t,0)}\right]}{dt}=\frac{dQ\left[\mathcal{E}_{(t,0)}\right]}{dt}=-\gamma(t)R(t)\log\left[\frac{1+R(t)}{1-R(t)}\right],
\end{equation}
which can be positive only for non-Markovian dynamics, $\gamma(t)<0$.
\end{ex}

Further examples of the use of these witnesses are found in \cite{ManiscalcoParisClassical} and \cite{Addis,Addis2} for a qubit interacting with a random classical field and with a bosonic environment, respectively.

\subsubsection{Bloch volume measure}

Another interesting proposal to witness non-Markovianity was suggested by Lorenzo, Plastina and Paternostro in \cite{Mauro}. These authors expand the state $\rho$ in the basis $\{G_j\}_{j=0}^{d^2-1} $ where $G_0=\mathbb{I}/\sqrt{d}$ and $\{G_j\}_{j=1}^{d^2-1}$ are the (normalized) generators of the $\mathfrak{su}(d)$ algebra,
\begin{equation}
\rho=\frac{\mathbb{I}}{d}+\sum_{j=1}^{d^2-1} r_jG_j, \quad r_i=\Tr(G_i\rho).
\end{equation}
Then, it is well-know that the action of a dynamical map can be seen as an affine transformation of the Bloch vector $\bm{r}=(r_1,\ldots, r_{d^2-1})^{\rm t}$,
\begin{equation}\label{affineBloch}
\rho(t)=\mathcal{E}_{(t,t_0)}(\rho)\longleftrightarrow \bm{r}_t=\bm{M}_{(t,t_0)}\bm{r}_{t_0}+\bm{c}_{(t,t_0)},
\end{equation}
where $[\bm{M}_{(t,t_0)}]_{ij}=\Tr[G_i\mathcal{E}_{(t,t_0)}(G_j)]$ and $[\bm{c}_{(t,t_0)}]_i=\Tr[G_i\mathcal{E}_{(t,t_0)}(\mathbb{I})]/d$ for $i,j>0$.

It can be proven that, since $\mathcal{E}_{(t,t_0)}$ is a composition of completely positive maps, the absolute value of the determinant of $\bm{M}_{(t,t_0)}$ decreases monotonically with time \cite{Wolf1}. Interestingly $|\det[\bm{M}_{(t,t_0)}]|$ describes the change in volume of the set of states accessible through the evolution \cite{Mauro}, so that Markovian evolutions reduce (or leave invariant) the volume of accessible states. Thus, this witness enjoys a nice geometrical interpretation. However, similarly to the trace distance, since the volume of accessible states is independent of the affine vector $\bm{c}_{(t,t_0)}$, it is not sensitive to dynamics where the non-Markovianity is encoded in $\bm{c}_{(t,t_0)}$. More concretely, it can be shown that the volume of accessible states only detects non-Markovian dynamics such that $\Tr[\bm{M}_{(t,t_0)}]>0$ \cite{MichaelHall}.

\begin{ex} For the pure dephasing model, Eq. \eqref{PureDephMaster}, we take the (normalized) generators of $\mathfrak{su}(2)$ algebra and the identity as basis, i.e. $\big\{\tfrac{1}{\sqrt{2}}\mathbb{I}_2,\tfrac{1}{\sqrt{2}}\sigma_x,\tfrac{1}{\sqrt{2}}\sigma_y,\tfrac{1}{\sqrt{2}}\sigma_z\big\}$. From \eqref{PureDephE} it is immediate to obtain that for this model $\bm{c}_{(t,0)}=0$ and
\begin{equation}
\bm{M}_{(t,t_0)}=\left[\begin{smallmatrix}
R(t) & 0 & 0  \\
0 & R(t) & 0   \\
0 & 0 & 1 &
\end{smallmatrix}\right],
\end{equation}
so that,
\begin{equation}
|\det[\bm{M}_{(t,t_0)}]|=R(t)^2,
\end{equation}
and we arrive to the same conclusions as with the quantum Fisher information Eq. \eqref{FisherEx}.
\end{ex}

For other examples of the use of $|\det[\bm{M}_{(t,t_0)}]|$ to track non-Markovianity see the original reference \cite{Mauro} and \cite{Mauro2ambientes,PlastinaFermionic,Tufarelli}.

\subsection{Witnesses based on monotonicity under local completely positive maps}

This second kind of witnesses are typically correlation measures between the dynamical system and some ancilla A, in such a way that they do not increase under the action of local maps, $\mathcal{E}\otimes\mathds{1}_{\rm A}$. Let us analyze three of them.

\subsubsection{Entanglement}\label{sec:entanglement}

From an operational point of view, entanglement can be defined as those correlations between different quantum systems which cannot be generated by local operations and classical communication (LOCC) procedures \cite{MartinShash}. Thus, entanglement turns out to be a resource to perform tasks which cannot be done just by LOCC.

The degree of entanglement of a quantum state may be assessed by the so-called {\it entanglement measures}. These must fulfill a set of axioms in order to account for the genuine properties present in the concept of entanglement \cite{MartinShash,MartinEnta1,MartinEnta2,HorodeckiReview}. One of these requirements is the monotonicity axiom, which asserts that the amount of entanglement cannot increase by the application of LOCC operations. Actually, the quantifiers of entanglement that fulfil this axiom but do not coincide with the entropy of entanglement for pure states are simply called {\it entanglement monotones}. Since local operations are a particular example of LOCC, if some entanglement measure (or monotone) increases under a local map $\mathcal{E}\otimes\mathds{1}$, $\mathcal{E}$ cannot be completely positive.

Thus, consider a system ${\rm S}$ evolving according to some dynamical map $\mathcal{E}_{(t,t_0)}$. We will study the evolution of an entangled state $\rho_{\rm SA}$ between S and some static ancillary system ${\rm A}$, $\rho_{\rm SA}(t)=\left[\mathcal{E}_{(t,t_0)}\otimes \mathds{1}\right](\rho_{\rm SA})$. Then, an increment in the amount of entanglement of $\rho_{\rm SA}(t)$ witnesses non-Markovianity. More specifically, consider initially a maximally entangled state $\rho_{\rm SA}=|\Phi\rangle\langle\Phi|$, $|\Phi\rangle=\frac{1}{\sqrt{d}}\sum_{n=0}^{d-1}|n\rangle_{\rm S} |n\rangle_{\rm A}$. Provided that $E$ is an entanglement measurement (or monotone), the positive quantity
\begin{equation}\label{noMarkoenta}
\mathcal{I}^{({\rm E})}:=\Delta E+\int_{t_0}^{t_1}\left|\frac{dE[\rho_{\rm SA}(t)]}{dt}\right|dt,
\end{equation}
is different from zero only if $\mathcal{E}_{(t,t_0)}$ is non-Markovian in the interval $(t_0,t_1)$. Here $\Delta E:=E[\rho_{\rm SA}(t_1)]-E[\rho_{\rm SA}(t_0)]$.

The use of the entanglement to witness non-Markovianity was first proposed in \cite{nuestro}, where the expression \eqref{noMarkoenta} was suggested. This proposal has been theoretical addressed for cases of qubits coupled to bosonic environments \cite{Br-Q-Bos3,Br-Q-Bos7,Br-Q-Bos15,Enta-1,Enta-2,Anomalous-non-Markovian}, for a damped harmonic oscillator \cite{nuestro,Doug,GalveZambrini}, and for random unitary dynamics and classical noise models \cite{Enta-3,Enta-4}. Experimentally this witness has been analyzed in \cite{ExperimentoPiilo1,RLoFranco,Enta-5}.

In addition, a link between the generation of entanglement by non-Markovian dynamics and the destruction of accessible information \cite{Vedral} as been established in \cite{AccessibleEntanglement}. Also, there is a relation between dynamical recovering of quantum coherence by applying local operations and entanglement generation, see \cite{DinamicalRecovering}. Further connections between entanglement and non-Markovianity can be found in \cite{Br-Q-Bos13,Enta-6}.

\subsubsection{Quantum mutual information} The total amount of correlations (as classical as quantum) as measured by the quantum mutual information is another witness of non-Markovianity \cite{Mutual1}. This quantity is defined as
\begin{equation}
I(\rho_{\rm SA}):=S(\rho_{\rm S})+S(\rho_{\rm A})-S(\rho_{\rm SA}),
\end{equation}
where $S(\rho):=-\Tr(\rho\log\rho)$ is the von Neumann entropy and $\rho_{\rm S,A}=\Tr_{\rm A,S}(\rho_{\rm SA})$. This expression can be rewritten as a relative entropy,
\begin{equation}
I(\rho_{\rm SA})=S(\rho_{\rm SA}\Vert \rho_{\rm S}\otimes\rho_{\rm A}).
\end{equation}
Thus, if we apply a local operation, we have
\begin{align}
I[(\mathcal{E}\otimes\mathds{1})\rho_{\rm SA}]&=S\{(\mathcal{E}\otimes\mathds{1})\rho_{\rm SA}\Vert \Tr_{\rm A}[(\mathcal{E}\otimes\mathds{1})\rho_{\rm SA}]\otimes\Tr_{\rm S}[(\mathcal{E}\otimes\mathds{1})\rho_{\rm SA}]\}\nonumber \\
&=S[(\mathcal{E}\otimes\mathds{1})\rho_{\rm SA}\Vert (\mathcal{E}\otimes\mathds{1})\rho_{\rm S}\otimes\rho_{\rm A}]\nonumber \\
&\leq S(\rho_{\rm SA}\Vert \rho_{\rm S}\otimes\rho_{\rm A})=I(\rho_{\rm SA}),
\end{align}
where have used that $\Tr_{\rm A}[(\mathcal{E}\otimes\mathds{1})\rho_{\rm SA}]=\mathcal{E}(\rho_{\rm S})$ and $\Tr_{\rm S}[(\mathcal{E}\otimes\mathds{1})\rho_{\rm SA}]=\rho_{\rm A}$, and the monotonicity of the relative entropy Eq. \eqref{QREmono}. Hence, the quantum mutual information is monotonic under local trace-preserving completely positive maps. For references where quantum mutual information has been used to study non-Markovianity see \cite{Mauro2ambientes,ComparisonBLPMutual,Mutual2,Anomalous-non-Markovian,LoFrancoPaladino,Addis2}.

\subsubsection{Quantum discord} Finally, the quantum discord \cite{ZurekAnnalen,Zurek,ModiReview} can also be used to detect non-Markovian dynamics. Recall that the quantum discord between two quantum systems is a non-symmetric measure of correlations, so that it is not the same to measure the quantum discord between S and A, with respect to A, as with respect to S. For our purposes, we consider the quantum discord as measured by the ancillary system:
\begin{equation}\label{discord}
D_{\rm A}^Q(\rho_{\rm SA}):=S(\rho_{\rm A})-S(\rho_{\rm SA})+\min_{\{\Pi_j^{\rm A}\}}\sum_j S\big(\rho_{S|\Pi_j^{\rm A}}\big),
\end{equation}
where the minimization is taken over all POVM $\{\Pi_j^{\rm A}\}$ on A, and $\rho_{S|\Pi_j^{\rm A}}$ is the system state after the outcome corresponding to $\Pi_j^{\rm A}$ has been detected,
\begin{equation}
\rho_{S|\Pi_j^{\rm A}}:=\frac{\Tr_{\rm A} (\Pi_j^{\rm A}\rho_{\rm SA})}{\Tr (\Pi_j^{\rm A}\rho_{\rm SA})}.
\end{equation}

The quantum discord Eq. \eqref{discord} is monotonic under local maps on the system $\mathcal{E}\otimes\mathds{1}$ (see \cite{ModiReview})
\begin{equation}
D_{\rm A}^Q[\mathcal{E}\otimes\mathds{1}(\rho_{\rm SA})]\leq D_{\rm A}^Q(\rho_{\rm SA}).
\end{equation}
However, note that this equation does not hold for local maps on the ancilla $\mathds{1}\otimes\mathcal{E}$ \cite{DiscordLocal1,DiscordLocal2}. Therefore, as long as we are certain that the ancilla does not evolve, we may use quantum discord to probe non-Markovianity.

The possible usefulness of the concept of quantum discord to witness non-Markovianity has been first discussed in \cite{Girolami}. Other relations between Markovianity and quantum discord can be found in \cite{Alipour,HaikkaYManiscalco,ZhangDiscord}. Note however that the existence of quantitative connections between quantum discord and completely positive maps remain controversial \cite{Cesar1,SL,Cesar2,NuestroCPDiscord,Sabapathy,Buscemi,DSL}.

\begin{ex} For the sake of illustration, let us analyze the behavior of entanglement, quantum mutual information and quantum discord in the pure dephasing model, Eq. \eqref{PureDephMaster}. Consider initially a maximally entangled state $\rho_{\rm SA}=|\Phi\rangle\langle\Phi|$, so that $\rho_{\rm SA}(t)=\left[\mathcal{E}_{(t,t_0)}\otimes \mathds{1}\right](|\Phi\rangle\langle\Phi|)$ is given by Eq. \eqref{Ex1PHI}.

As an entanglement monotone we may take the logarithmic negativity \cite{MartinShash,MartinNegativity}, arriving at
\begin{equation}
E_N[\rho_{\rm SA}(t)]=\log\|\rho_{\rm SA}^{{\rm t}_A}(t)\|_1=\log[1+R(t)],
\end{equation}
where the superscript ${\rm t}_A$ denotes the partial transpose with respect to the ancillary system. Immediately we obtain
\begin{equation}
\frac{dE_N[\rho_{\rm SA}(t)]}{dt}=-2\gamma(t)\frac{R(t)}{1+R(t)},
\end{equation}
so entanglement can increase only for non-Markovian evolution $\gamma(t)<0$.

In this simple model, the quantum mutual information and quantum discord are found to be
\begin{align}
I[\rho_{\rm SA}(t)]&=C_{ea}\left[\mathcal{E}_{(t,0)}\right],\\
D_{\rm A}^Q[\rho_{\rm SA}(t)]&=Q\left[\mathcal{E}_{(t,0)}\right],
\end{align}
where $C_{ea}$ and $Q$ are given in Eqs. \eqref{CeaEq} and \eqref{QEq}. The equality follows from the choice of the maximally entangled state $\rho_{\rm SA}=|\Phi\rangle\langle\Phi|$ as initial state. This is a purification of the maximally mixed state, which is the one that solves the optimization problem in Eqs. \eqref{CeaDefi} and \eqref{QDefi} for this model as commented in Example \ref{ExCapacity}. For the quantum mutual information the equality is obvious. For the quantum discord, note that the state \eqref{Ex1PHI} belongs to the subclass known as ``X-states'', for which the optimization problem in Eq. \eqref{discord} can be efficiently solved \cite{ModiReview,Chen}. Concretely in this case, we take the measurement of the $\sigma_z$ observable to obtain $\rho_{S|\Pi_{\pm1}^{\rm A}}=|z_{\pm}\rangle\langle z_{\pm} |$, where $\sigma_z|z_{\pm}\rangle=\pm|z_{\pm}\rangle$, and so the von Neumann entropy of any system state after the measurement vanishes.
\end{ex}

\section{Conclusion and Outlook}
In this work we have reviewed the topic of quantum non-Markovianity in the light of recent developments regarding its characterization and quantification. Quantum Markovian processes have been defined by taking the {\em divisibility} approach, which allows us to circumvent the problem of constructing a hierarchy of probabilities in quantum mechanics. We have also discussed the emergence of memorylessness properties within this definition and compared the divisibility approach with other suggested ways to define Markovianity in the quantum realm. We have surveyed recently proposed measures and witnesses of non-Markovianity, explaining their foundations, as well as their motivation and interpretation. Each measure and witness of non-Markovianity has its pros and cons, and the ultimate question of which of them is preferable in practice strongly depends on the context.

We hope that this article can be useful for future research in open quantum systems, and its implications for other areas such as quantum information, or statistical mechanics. Despite the tremendous quantity of new results in the characterization and quantification of non-Markovianity in recent years, there are still several important open questions that remain to be addressed. We conclude this review by providing a non-exhaustive list and some possible research directions.

\begin{enumerate}[$\blacktriangleright$]
\item \textit{Classification of completely positive non-Markovian master equations}. This is probably the most general open problem regarding non-Markovian evolutions. For instance, in Eq. \eqref{ACHgenerator} we have written a generic master equation without taking care of the completely positive character of the dynamics that it generates. When the evolution is non-Markovian, the structure of the generators which leads to completely positive dynamics is pretty much unknown, although a few partial results have been obtained \cite{Hall,t0definition2,Commutative}. The problem basically rests upon the difficult characterization of the generators of dynamical maps $\{\mathcal{E}_{(t_2,t_1)},t_2\geq t_1\geq t_0\}$ under the weak assumption of complete positivity just for instants $t_2\geq t_1=t_0$, and not for any $t_2\geq t_1\geq t_0$ \cite{Libro}.
\item \textit{Computation of some measures of non-Markovianity}. Despite their well grounded physical motivation, it would be desirable to provide efficient ways to compute some of the proposed measures of non-Markovianity. For instance, the geometric degree of non-Markovianity, $\mathscr{D}_{\rm NM (g)}^I$, Eq. \eqref{geoDegree}. Similarly, the measure $\mathcal{N}_{\rm H}^I$, Eq. \eqref{NHelstrom}, has been calculated just for unbiased problems or isolated cases without solving the complete optimization problem.
\item \textit{Performance of witnesses}. Some questions may be posed regarding the performance of witnesses of non-Markovianity. For example, which kind of witnesses, be it the ones based on monotonicity or the ones based on local monotonicity, is more sensitive to non-Markovian dynamics? Moreover, which of them is more efficient to detect non-Markovian dynamics? A recent study with some partial results on this issue is \cite{Addis2}.
\item \textit{Witnesses of non-Markovianity without resorting to full-state tomography}. A question of practical interest is to formulate ways to probe non-Markovianity avoiding full-state (or process) tomography. For example, if we manage to find good enough lower and upper bounds to properties like trace distance \cite{NC00} or entanglement \cite{Audenaert} in terms of simple measurements, we would be able to detect its non-monotonic behavior without resorting to expensive tomographic procedures.
\item \textit{Relation between different measures of non-Markovianity}. Another fundamental question is to elucidate whether different measures of non-Markovianity induce the same order. Probably the answer is negative, but more progress has to be done on this line.
\item \textit{Non-Markovianity as a resource theory}. Related to the previous point is the possible formulation of a resource theory for non-Markovianity. Similarly to other resource theories \cite{MartinShash,MartinEnta1,MartinEnta2,HorodeckiReview,Brandao1,Gour1,Brandao2,Emerson,Gour2,Tilmann,deVicente}, we may wonder if non-Markovianity can be seen as a resource to perform whatever tasks which cannot be done solely by Markovian evolutions. Then, an order relation follows, i.e. some dynamics $\mathcal{E}^{(1)}_{(t,t_0)}$ has smaller amount of Markovianity than another dynamics $\mathcal{E}^{(2)}_{(t,t_0)}$, if $\mathcal{E}^{(1)}_{(t,t_0)}$ can be constructed by $\mathcal{E}^{(2)}_{(t,t_0)}$ and Markovian evolutions. This approach also allows to introduce the notion of maximally non-Markovian evolutions, these would be the ones which cannot be generated in terms of other non-Markovian maps and Markovian evolutions. Would this maximally non-Markovian evolutions be the ones defined in Section \ref{sec:DariuszSabrina}?
\item \textit{Non-Markovianity and other properties}. It will be very relevant to find systems where the presence of non-Markovianity is associated with other notable phenomena. For instance, some quantitative relations between non-Markovianity effects and criticality and phase transitions \cite{Loschmidt2,Sindona,Borrelli,Borrelli2}, Loschmidt echo \cite{Loschmidt,Loschmidt2,Sindona}, symmetry breaking \cite{Chancellor}, and Zeno and anti-Zeno effects \cite{NoMarkovZeno} have already been described.
\item \textit{Potential applications of non-Markovianity}. As final question, we may wonder is {\em ``what is a non-Markovian process good for in practice?''} There are already studies showing its usefulness to prepare steady entangled states \cite{HRP}, to enhance the achievable resolution in quantum metrology \cite{AlexSusanaMartin} or to assist certain tasks in quantum information and computation \cite{ElsiTeleport,DarekSabrina,PiiloEntaDistri}. However more research on this direction is required for the formulation of quantitative results \cite{Anomalous-non-Markovian}.
\end{enumerate}

Definitely, this list can be extended with other open questions. However, we think the enumerated points are representative enough to hopefully stimulate the readers into addressing some of these problems and shed further light on this remarkable phenomenon.

\section*{Acknowledgements}
A.~R. acknowledges to J.~M.~R. Parrondo for discussions about the classical definition of Markovianity and to L. Accardi and K.~B. Sinha for their detailed explanations about the algebraic definition of stochastic quantum processes. Moreover, it has been a pleasure to share ideas about quantum non-Markovianity with D. Chru\'sci\'nski, A. Kossakowski, M.~M. Wolf, F. Ticozzi, M. Paternostro and M. J. W. Hall. We acknowledge financial support from Spanish MINECO grant FIS2012-33152, CAM research consortium QUITEMAD S2009-ESP-1594, UCM-BS grant GICC-910758, EU STREP project PAPETS, the EU Integrating project SQIS, the ERC Synergy grant BioQ and an Alexander von Humboldt Professorship.

\end{document}